%% file: TACAS.tex
\documentclass[runningheads]{llncs}
\usepackage{lmodern}
\usepackage[T1]{fontenc}
\usepackage{graphicx}
\usepackage{amsmath}
\usepackage[capitalize]{cleveref}
\usepackage{multirow}
\usepackage{graphicx}
\usepackage{multicol}
\usepackage[title]{appendix}
\input{bb_common-citations}

\input{local-macros}
\input{bb_common-tex-defs}

\newcommand\doublBlind[1]{}

\begin{document}
\title{Incremental SAT-Based Enumeration of Solutions to the Yang-Baxter Equation}

\titlerunning{SAT-Based Enumeration of Solutions to the YBE}

\author{Daimy Van Caudenberg\inst{1}\orcidID{0000-0002-7975-4838}\and
Bart Bogaerts\inst{1,2}\orcidID{0000-0003-3460-4251} \and
Leandro Vendramin\inst{3}\orcidID{0000-0003-0954-7785}}

\authorrunning{D. Van Caudenberg et al.}

\institute{KU Leuven, Dept.\ of Computer Science; Leuven.AI, 
	B-3000 Leuven, Belgium
	\\\email{Daimy.VanCaudenberg,Bart.Bogaerts@kuleuven.be}\\\and
Vrije Universiteit Brussel,  Dept.\ of Computer Science, B-1050 Brussels, Belgium \\
\and Vrije Universiteit Brussel,  Dept.\ of  Mathematics and Data Science,  B-1050 Brussels, Belgium\\
\email{Leandro.Vendramin@vub.be}}

\maketitle

\begin{abstract}
  We tackle the problem of enumerating set-theoretic solutions to the Yang-Baxter equation.
	This equation originates from statistical and quantum mechanics, but also has applications in knot theory, cryptography, quantum computation and group theory.
	Non-degenerate, involutive solutions have been enumerated for sets up to size 10 using constraint programming with partial static symmetry breaking~\cite{AMV22Enumerationset-theoreticsolutionsYang-Baxterequation}; for general non-involutive solutions, a similar approach was used to enumerate solutions for sets up to size 8.
	In this paper, we use and extend the SAT Modulo Symmetries framework (SMS), to expand the boundaries for which solutions are known.
	The SMS framework relies on a \emph{minimality check}; we present two solutions to this, one that stays close to the original one designed for enumerating graphs and a new incremental, SAT-based approach.
	With our new method, we can reproduce previously known results much faster and also report on results for sizes that have remained out of reach so far.
	This is an extended version of a paper to appear in the proceedings of the 31st International Conference on Tools and Algorithms for the Construction and Analysis of Systems.
\end{abstract}

\renewcommand\bart[1]{\relax}

\section{Introduction}
\label{sec:introduction}

\subsubsection{The Yang-Baxter Equation}
The Yang--Baxter equation (YBE) is a fundamental equation in mathematical physics,
with origins in the study of statistical mechanics. Its
significance, however, has quickly extended to various areas of pure mathematics, including
representation theory and low-dimensional topology. More recently, a discrete
version of the equation has become central in algebra and combinatorics.
The set-theoretic solutions to this discrete
version of the equation were first highlighted by Drinfeld \cite{zbMATH04091684},
who emphasized the importance of discovering new solutions and the inherent
difficulty of this task. He suggested that set-theoretical solutions, in
particular, merited further investigation. 
The study and construction of such solutions
began with
a series of pioneering papers \cite{zbMATH01214144,zbMATH01425260,zbMATH01501663,zbMATH01585085}.
These works apply tools from ring theory, group theory and homological algebra 
to analyze non-degenerate solutions, as these solutions naturally have
groups acting on them. Moreover, these papers demonstrated that these 
solutions offer a particularly rich source of examples, with numerous
applications in algebra \cite{zbMATH06289322,zbMATH06713497,zbMATH05118810} 
and connections with other topics, 
such as Hopf--Galois structures \cite{zbMATH07453873}.

Constructing finite solutions of small size is particularly intriguing for
several reasons. Firstly, it presents a highly challenging problem at the
intersection of computational techniques and algebra. Secondly, 
solutions provide, for example, concrete colouring invariants of knots \cite{zbMATH05998842} and
intriguing examples of algebraic structures to study \cite{arXiv:2305.06023}. Thirdly,
examples of solutions allow for experimentation, which can reveal patterns that
provide deeper insights into the structure of combinatorial solutions.
Lastly, the ubiquity of the Yang--Baxter equation suggests that the methods used to construct its solutions could be adapted to other contexts, even those not directly related to the equation itself. 
For example, these ideas can be applied to the problem of constructing $L$-algebras \cite{zbMATH05365168}, something that has direct applications in algebraic logic.

\subsubsection{SAT Modulo Symmetries} 
As usual in mathematics, one is not interested in generating \emph{all} solutions, but only all non-isomorphic solutions.
One prominent technique for isomorph-free enumeration is SAT Modulo Symmetries (SMS)~\mycite{SMS}.
Given a formula in conjunctive normal form (CNF) and a possibly empty set of symmetries, the goal of SMS is to find all satisfying assignments that are lexicographically minimal among their symmetric counterparts (here, symmetries of the CNF correspond to isomorphisms of the original problem).
To do this, SMS uses a minimality check that takes place \emph{during} the solving phase.
In other words, given a (possibly) partial assignment, it is verified whether that assignment can still be extended to a complete, lexicographically minimal assignment.

When the underlying problem is known, the minimality check can take into account that not all complete assignments are considered, but only the assignments that model the formula.
Hence, the minimality check can be optimized using domain-specific knowledge.
This has been used for enumerating graphs~\mycite{SMS}, matroids~\cite{KSS22SATAttackRotasBasisConjecture} and so-called \emph{Kochen-Specker vector systems}~\cite{KPS23Co-CertificateLearningSATModuloSymmetries}.

\subsubsection{Our contributions}
Inspired by the success of SMS, we now extend it to enumerate non-isomorphic solutions to the Yang-Baxter equation.
First, we discuss how to encode the YBE in CNF.
The most challenging aspect is then to develop a domain-specific minimality check, taking into account the specific problem structure.
We can then use this minimality check to learn clauses that are added \emph{during} the solving phase to eliminate parts of the search tree that contain no lexicographically minimal solutions.

We describe and test two implementations of this minimality check. 
The first is a backtracking approach that stays close to the original SMS work, also making use of clever representations of sets of symmetries.
The second one is a new approach based on incremental SAT solving where the \mbox{(non-)minimality} problem \emph{itself} is encoded as a CNF and a SAT oracle is repeatedly used to search for symmetries that guarantee that the solution at hand is non-minimal. 
Both approaches result in quite significant speed-ups compared to the previous approach~\cite{AMV22Enumerationset-theoreticsolutionsYang-Baxterequation}; experiments show that the incremental, SAT-based minimality check outperforms the backtracking approach.
Moreover, with this new approach, we can report on results for sizes that have remained out of reach so far.

\section{Preliminaries}
\label{sec:prelims}

\subsubsection{Yang-Baxter Equation and Cycle Sets}

The \emph{Yang-Baxter equation} in general is concerned with vector spaces and
mappings between them. In this paper, we focus on set-theoretic solutions
to this equation, as introduced by Drinfeld~\cite{D92someunsolvedproblemsquantumgrouptheory}, and particularly on
involutive non-degenerate solutions. 
These objects are known to be equivalent to \emph{non-degenerate cycle sets}~\cite{MR2132760}, which form a much simpler class of combinatorial objects.
We omit the definition of the
Yang-Baxter equation in full generality and immediately define cycle sets; we
refer the reader for example to Akg\"un, Mereb and Vendramin~\cite{AMV22Enumerationset-theoreticsolutionsYang-Baxterequation} for more details. 

\newcommand\binop{\cdot}
\newcommand\leftmult{\phi} 
\newcommand\set{X}
\begin{definition}\label{def:cycle_set}
	A cycle set $(\set,\binop)$ is a pair consisting of a non-empty set $\set$ and a binary operation~$\binop$ on $\set$ that fulfills the following relations:
	\begin{enumerate}
		\item the map $\leftmult_x:\set\to\set:y\mapsto x\binop y$ is bijective for all $x\in\set$ and
		\item for all $x,y,z\in\set$: 
		\begin{equation*}\label{eq:cycloid}
			(x\binop y)\binop(x\binop z)=(y\binop x)\binop(y\binop z). \text{\qquad(the cycloid equation)}
		\end{equation*}
	\end{enumerate}
It is called \emph{non-degenerate} if the map $\set\to\set:x\mapsto x\binop x$ is bijective. 
\end{definition}

The rest of this paper is concerned with computer-aided enumeration of all non-degenerate cycle sets of a given (finite) size. 

\newcommand\assignment{\alpha} 
\newcommand\formula{F} 

\subsubsection{Boolean Satisfiability (SAT)}
A \emph{(Boolean) variable} takes values in $\{\ltrue,\lfalse\}$; a \emph{literal} is a variable $x$ or its negation $\olnot x$. 
A \emph{clause} is a disjunction of literals and a \emph{formula} (in conjunctive normal form) is a conjunction of clauses. 
A \emph{(partial) assignment} is a consistent set of literals (i.e., a set of literals that does not contain a literal and its negation). 
An assignment is \emph{complete} (for formula $\formula$) if it contains either $x$ or $\olnot x$ for each variable $x$ occurring in $F$. 
An assignment $\assignment$ \emph{satisfies} a formula $\formula$ if it contains at least one literal from each clause in $F$. 
The \emph{SAT problem} consists of deciding for a given formula $\formula$ whether an assignment exists that satisfies it.

\subsubsection{Symmetry in SAT} Symmetry has a long history in SAT, with various tools being used to exploit them either before search (\eg \cite{\refto{shatter},DBBD16ImprovedStaticSymmetryBreakingSAT,ABR24SatsumaStructure-BasedSymmetryBreakingSAT}) or during the search  (\eg \cite{S09SymChaffexploitingsymmetrystructure-awaresatisfiabilitysolver,BNOS10EnhancingClauseLearningSymmetrySATSolvers,DBB17SymmetricExplanationLearningEffectiveDynamicSymmetry,\refto{SMS}}). 
The \emph{SAT Modulo Symmetries (SMS)} framework \mycite{SMS}  stands out in this list by going a step further than just deciding whether a formula is satisfiable; instead, the goal is to enumerate assignments that satisfy it. 
It was designed with use cases in mathematics in mind and in particular it was first used to enumerate graphs that have certain interesting properties. 
In such use cases, one is often not interested in \emph{all} graphs that satisfy these properties, but only in generating non-isomorphic graphs. 
The core idea underlying SMS is that we can \emph{(1)} encode as a propositional formula what it means to be a graph that satisfies these interesting properties (or more generally, a suitable mathematical structure) and \emph{(2)} force a SAT solver, during the search, to generate \emph{canonical representations} of each of the classes of isomorphic solutions. 
The second point is achieved by designing a procedure that is aware of the isomorphisms of the problem at hand, known as the \emph{minimality check}. 
This takes the state of the SAT solver (a partial interpretation) as input and checks whether it can still be extended to a complete assignment that represents a lexicographically minimal graph (among all graphs isomorphic to it). If not, it forces the solver to abort the current branch of the search tree by analyzing \emph{why} this is no longer possible and learning a new clause from it that is then added to the solver's working formula. 
In general, this minimality check is incomplete but guaranteed to be complete when run on complete assignments.
This minimality check needs to be designed for each application, taking into account the (encoding of) the mathematical problem at hand as well as the structure of the set of isomorphisms. 
Seress~\cite{alma9992322369701471} gives a detailed introduction on group theory, for a lighter intoduction we refer readers to Gent et al.~\cite{GPP06SymmetryConstraintProgramming}.

\section{SAT Modulo Symmetries for the Yang-Baxter Equation}

We now explain how we use the SMS framework to enumerate cycle sets of a given size. 
To do this we first discuss how to construct a propositional formula that encodes the properties of cycle sets.
Next, we discuss the isomorphisms of the problem, and how the minimality check is adapted to deal with these isomorphisms.

\newcommand\ybemat{C\xspace}
\newcommand\matvar{v\xspace}
\newcommand\lvar{l\xspace}
\newcommand\rvar{r\xspace}
\newcommand\yvar{y\xspace}

\subsection{SAT Encoding} 
\label{ss:satenc}
\label{enc}
A cycle set of size $n$ can be viewed as an $n\times n$ matrix taking at each entry a number between $1$ and $n$ and satisfying some structural properties. 
This matrix represents the binary operation of the cycle set: at position $i,j$, the value is $i\cdot j$. 
In our algorithms, we will often take the view of a cycle set being such an integer-valued matrix, while the actual SAT encoding, obviously, talks about lower-level variables. 
The properties our matrix $\ybemat\in\set^{n\times n}$ should satisfy are: 
\begin{enumerate}
 \item \label{prop:bij} for all $x,y,z\in\set$ with $y\neq z$, $\ybemat_{x,y}\neq\ybemat_{x,z}$ ($\leftmult_x$ is bijective for all $x\in\set$),
 \item \label{prop:cycloid} for all $x,y,z\in\set$, $\ybemat_{\ybemat_{x,y},\ybemat_{x,z}}=\ybemat_{\ybemat_{y,x},\ybemat_{y,z}}$ (the cycloid equation holds), and
 \item \label{prop:nondeg} for all $x,y\in\set$ with $x\neq y$, $\ybemat_{x,x}\neq\ybemat_{y,y}$ (the cycle set is non-degenerate).
\end{enumerate}

To encode these properties in CNF, we will make use of the one-hot encoding, \ie for each $i,j,k\in\set$ we introduce a Boolean variable $\matvar_{i,j,k}$ which is true if and only if $\ybemat_{i,j}=k$.
We then add clauses stating that for each cell, exactly one of its indicator variables must hold;
\begin{align}
	&\texttt{ExactlyOne}(\{ \matvar_{i,j,k}\mid k\in\set \}), &&\text{(for each $i,j\in \set$)}
\end{align}
where $\texttt{ExactlyOne}$ refers to a set of clauses that may use auxiliary variables and that enforce that exactly one of its input variables holds.  
These clauses are constructed using either a binary encoding or the commander encoding \cite{KK07EfficientCNFEncodingSelecting1from}.
We then encode the properties in clauses over these variables, ensuring that each number should occur exactly once on each row and the diagonal.
Last, we also ensure that the cycloid equation holds. 
To do this, for each $i,j,k,b\in X$, we introduce a new variable $\yvar_{i,j,k,b}$ that is true precisely when $\ybemat_{\ybemat_{i,j},\ybemat_{i,k}} = b$ and that is also true precisely when $\ybemat_{\ybemat_{j,i},\ybemat_{j,k}} = b$.
The complete encoding can be found in Appendix~\ref{app:propCNF}.

\newcommand\symmgroup{\mathcal{S}\xspace}
\subsubsection{Fixing the Diagonal}
\label{ss:diagfix}
The problem defined in the previous section can be solved as is, however, the search space can be significantly reduced using a trick introduced by Akg\"un, Mereb and Vendramin~\cite{AMV22Enumerationset-theoreticsolutionsYang-Baxterequation}.
In short, if the diagonals of two cycle sets are conjugates, we know that the cycle sets are isomorphic.
If the goal is to only enumerate \emph{non-isomorphic} cycle sets, we can safely partition the problem by enumerating the solutions for one diagonal per conjugacy class.
Given a cycle set $(\set,\binop)$, its diagonal can be expressed as a permutation of the elements in $\set=\{1,\ldots,n\}$.
In other words, the set of possible diagonals is given by the symmetric group over $\set$, i.e., $\symmgroup_n$.
It is known that $\symmgroup_n$ has $p(n)$ conjugacy classes, where $p(n)$ is the number of integer partitions of $n$.
Hence, the search space can be drastically reduced by partitioning the problem into $p(n)$ problems with the diagonal fixed to a representative of its conjugacy class.

Given a fixed diagonal (i.e., values for all matrix cells $\ybemat_{i,i}$), we fix these variables in our encoding by only introducing variables $\matvar_{i,j,k}$ whenever if $j\neq i$ and $k\neq \ybemat_{i,i}$ and simplifying the rest of the theory accordingly. 
Fixing the diagonal does not only allow us to simplify the encoding, later we will show that it also allows us to optimize the minimality check.
\newcommand\obinop{\times}
\newcommand\oset{Y}
\newcommand\perm{\pi}
\newcommand\invperm{\pi^{-1}}
\newcommand\centr{C}
\newcommand\leqC{\preceq}
\newcommand\lneqC{\prec}
\newcommand\leqP{\trianglelefteq}
\newcommand\lneqP{\vartriangleleft}
\newcommand\leqCc{\leq_{\cycsets}}
\newcommand\lneqCc{<_{\cycsets}}
\newcommand\leqPc{\leqP}
\newcommand\lneqPc{\lneqP}
\newcommand\gneqP{\succ_{\pcycsets}}
\newcommand\geqP{\succeq_{\pcycsets}}
\newcommand\tuple[1]{\langle#1\rangle}
\newcommand\pcycsets{\mathcal{P}_n}
\newcommand\cycsets{\mathcal{C}_n}

\subsection{Isomorphisms and Symmetries}
Two cycle sets $(\set,\binop)$ and $(\set,\obinop)$ are isomorphic if and only if there exists an isomorphism $\perm:\set\to\set$ such that $\perm(x\binop y)=\perm(x)\obinop\perm(y)$.
For the matrices $\ybemat$ and $\ybemat'$ corresponding to $(\set,\binop)$ and $(\set,\obinop)$ respectively, this means that $\ybemat'=\perm(\ybemat)$ where $\perm(\ybemat)_{i,j}=\invperm(\ybemat_{\perm(i),\perm(j)})$.
When enumerating all cycle sets of a fixed size, we are only interested in enumerating non-isomorphic solutions;  
we choose to enumerate cycle sets that are lexicographically minimal (among all isomorphic solutions).
In other words, we are only interested in those cycle sets $(\set,\binop)$, whose associated matrices are lexicographically smaller than or equal to those of their isomorphic variants.
Concretely, given a cycle set $(\set,\binop)$, with associated matrix $\ybemat$, we need the following to hold:
\[\forall\perm\in\mathcal{I}:\ybemat\leqC\perm(\ybemat),\]
where $\leqC$ is the lexicographical ordering given by $\ybemat\leqC\ybemat'$ if $C=C'$ or  $C_c < C'_c$ for the 
first\footnote{Given a fixed order of the cells.} cell $c$ where $C$ and $C'$ differ,  
and $\mathcal{I}$ is the set of isomorphisms of the problem.
However, it is important to note that these isomorphisms may not correspond to symmetries of the propositional formula.
Indeed, for our specific use case, the commander encoding of the $\texttt{ExactlyOne}$ constraint introduces auxiliary variables that break some of the symmetries. 
As a consequence, symmetries cannot be detected from the CNF encoding but should really be determined by the problem at hand. 

In the general case, all isomorphisms $\perm\in\symmgroup_{n}$ with $n=|\set|$ need to be considered. However, when working with a fixed diagonal, not all $\perm\in\symmgroup_n$ are isomorphisms of the problem, but only the permutations that fix the diagonal.
Specifically, the isomorphisms of a problem with fixed diagonal $T$ consist of the set $\centr_{\symmgroup_n}(T)$, i.e., the centralizer of $T$ in the symmetric group.
These are exactly the elements $\perm\in\symmgroup_n$ for which it holds that $\perm\circ T=T\circ\perm$.

\subsection{Minimality Check}
We now discuss the core of the SAT Modulo Symmetries framework, namely the minimality check, the goal of which is deciding whether the current 
assignment (which can be thought of as a partially constructed matrix) can still be extended to a lexicographically minimal matrix (among its symmetric images).

\newcommand\pybemat{P}
\newcommand\asg{\assignment}
\newcommand\xcycsets[1]{\mathcal{X}(#1)}

\subsubsection{Partial Cycle Sets}
A \emph{partial cycle set} is a matrix $\pybemat\in{(2^\set)}^{n\times n}$ with $n=|\set|$, where each cell $c\in\set\times\set$ of the matrix represents a non-empty \emph{domain} $\pybemat_c\subseteq X$ of values that are still possible.
Given a (partial) assignment $\asg$, we can extract a partial cycle set $\pybemat^\asg$ as follows:
\[\forall c\in\set\times\set:\pybemat^\asg_c=\{x\in\set\mid\olnot\matvar_{c,x}\not\in\asg\},\]
where if $c=\tuple{i,j}$, we denote $\matvar_{i,j,x}$ as $\matvar_{c,x}$.
In other words, the partial cycle set will contain only values that can still be true according to the assignment. 

Given a partial cycle set $\pybemat$ and cell $c$, we say that the entry $\pybemat_c$ is \emph{defined} if it equals a singleton $\{x\}$, and write $\pybemat_c=x$.
In all other cases, we say that $\pybemat_c$ is \emph{undefined}.
A complete cycle set $\ybemat$ is a partial cycle set with no undefined values and for which the cycle set constraints discussed in \cref{ss:satenc} hold.
In this case, we associate $\ybemat$ with an actual cycle set.
The set $\pcycsets$ denotes all partial cycle sets, similarly, the set of all complete cycle sets is denoted $\cycsets$.
A partial cycle set $\pybemat\in\pcycsets$ can be extended to another partial (or complete) cycle set $\pybemat'\in\pcycsets$, if for all cells $c$ it holds that $\pybemat'_c\subseteq\pybemat_c$.
Given a partial cycle set $\pybemat\in\pcycsets$, we denote the set of complete cycle sets that $\pybemat$ can be extended to by $\xcycsets{\pybemat}$.
\begin{example}
Given a partial cycle set
\[\pybemat=
\left[\begin{matrix}
$\{2\}$             & $\{1\}$           &  $\{3\}$ \\
$\{2\}$           & $\{1\}$           &  $\{3\}$ \\
$\{1,2\}$         & $\{1,2\}$         &  $\{3\}$
\end{matrix}
\right], \text{ we have that }
\xcycsets{\pybemat}=
\left\{
\left[\begin{matrix}
2 & 1 & 3\\
2 & 1 & 3\\
1 & 2 & 3
\end{matrix}
\right],
\left[\begin{matrix}
2 & 1 & 3\\
2 & 1 & 3\\
2 & 1 & 3
\end{matrix}
\right]
\right\}.
\]
\end{example}

The goal of the minimality check is, given a partial cycle set, to determine whether it is possible to extend it to a lexicographically minimal complete cycle set. One way to achieve this would be to loop over all (exponentially many) such extensions and check for minimality. 
This would, however, clearly not be feasible. 
Instead, inspired by and building on the work of the original SMS paper~\cite{KS21SATModuloSymmetriesGraphGeneration}, we take a different approach, and we will search for a symmetry $\perm$ (in our case, this is a permutation  of $X$ that respects the fixed diagonal) that guarantees that for every complete extension $\ybemat$ of $\pybemat$, 
\[\perm(C)\prec C. \]

We will call such a permutation a \emph{witness of non-minimality}, and the goal of the minimality check is to find such \emph{witnesses}.

Let us describe a sufficient condition for finding such witnesses that forms the basis of our algorithms.
Given $\pybemat\in\pcycsets$ and a permutation $\perm$, we can apply $\perm$ to $\pybemat$ as follows; for each cell $c=\tuple{i,j}$, we define the image of $c$ under $\perm$ as $\perm(c)=\tuple{\perm(i),\perm(j)}$.
Given a non-empty domain $S\subseteq\set$, we define $\perm(S)=\{\perm(x)\mid x\in S\}$.
Using this, we define $\perm(\pybemat)$ as the partial cycle set where $\perm(\pybemat)_c=\invperm(\pybemat_{\perm(c)})$ for each cell $c$. 
Note that for complete cycle sets, this coincides with the definition of $\perm(C)$ given above. 
We also define a lexicographical order $\leqP$ over partial cycle sets, with the specific characteristic that if $\pybemat\lneqP\pybemat'$ then for all extensions $\ybemat\in\xcycsets{\pybemat}$ and $\ybemat'\in\xcycsets{\pybemat'}$ it holds that $\ybemat\lneqC\ybemat'$.
Hence, if the minimality check can find a permutation $\perm$ such that $\perm(\pybemat)\lneqP\pybemat$, we have that $\perm(\ybemat)\lneqC\ybemat$ for all extensions $\ybemat\in\xcycsets{\pybemat}$, i.e., that $\perm$ is a witness of non-minimality.
In order to define $\lneqP$, we note that $\pybemat\lneqP\pybemat'$ should guarantee that 
\begin{itemize}
 \item there is a cell $c$, the value of which in $\pybemat$ is guaranteed to be strictly smaller than that of  $\pybemat'_c$, and,
 \item for all cells $c'<c$, the value of $c'$ in  $\pybemat$ is guaranteed to be at most the value it takes in  $\pybemat'$.
\end{itemize}
Since we want to compare the value of cells where the value is potentially not determined yet, we extend the order on values to domains as follows. 
\begin{definition}
Let $S$ and $S'$ be two non-empty subsets of $\set$, we define 
	\begin{itemize}
		\item $S\leqPc S'$ if and only if $\max{S}\leq\min{S'}$ and
		\item $S\lneqPc S'$ if and only if $\max{S}<\min{S'}$.
	\end{itemize}
\end{definition}

\begin{definition}
	Let $\pybemat$ and $\pybemat'$ be two partial cycle sets and let $c$ be a cell.
	We say that \emph{$\pybemat$ is below $\pybemat'$ up to cell $c$} (and denote this $\pybemat\leqP_c \pybemat'$) if for all cells $c'\leq c$ it holds that $\pybemat_{c'}\leqPc\pybemat'_{c'}$. 
\end{definition}

\begin{definition}\label{def:lneqP}
	 We say $\pybemat$ is \emph{strictly smaller} than $\pybemat'$ (and denote this $\pybemat\lneqPc\pybemat'$) if there is a cell $c$ such that 
\begin{itemize}
 \item $\pybemat\leqP_{\mathrm{pred}(c)} \pybemat'$, where $\mathrm{pred}(c)$ is the cell immediately preceding $c$ in the lexicographic ordering, and
 \item $\pybemat_c\lneqPc\pybemat'_c$.
\end{itemize}
We say that $\pybemat$ is \emph{at most} $\pybemat'$ (and denote this $\pybemat\leqP\pybemat'$) if either $\pybemat \lneqPc \pybemat'$ or if for all cells $c$, $\pybemat_c\leqPc\pybemat'_c$.
\end{definition}

These definitions guarantee strong properties of complete extensions of $\pybemat$ and $\pybemat'$. 
Proofs of our results can be found in Appendix \ref{app:proofs}.
\begin{proposition}\label{prop:prop}
	Let $\pybemat$ and $\pybemat'$ be two partial cycle sets. The following properties hold. 
\begin{enumerate}
 \item \label{prop:fst} If $\pybemat\lneqP\pybemat'$, then for all  $\ybemat\in\xcycsets{\pybemat}$ and $\ybemat'\in\xcycsets{\pybemat'}$, it holds that $\ybemat\lneqC\ybemat'$.
 \item \label{prop:snd} If $\pybemat\leqP\pybemat'$, then for all  $\ybemat\in\xcycsets{\pybemat}$ and $\ybemat'\in\xcycsets{\pybemat'}$, it holds that $\ybemat\leqC\ybemat'$. 
 \item \label{prop:thrd} If $\pybemat\leqP\pybemat'$ and $\pybemat'\leqP\pybemat$, then $\pybemat=\pybemat'$ and both cycle sets are complete.
\end{enumerate}
\end{proposition}
This now gives rise to a sufficient condition for being a witness of non-minimality (which allows us to conclude that there are no lexicographically minimal extensions of $\pybemat$) or that some values of a cell are impossible (which would allow us to propagate extra information during search). 
\begin{theorem}\label{thm:main}
 Given a partial cycle set $\pybemat$ and permutation $\perm$, we have that:
 \begin{itemize}
  \item if $\perm(\pybemat)\lneqP\pybemat$, then $\perm$ is a witness of non-minimality, 
  \item If $\perm(\pybemat)\leqP_c\pybemat$ then
  there is an extension $\pybemat'$ of $\pybemat$ such that
  \begin{enumerate}
  	\item $\pybemat'$ and $\perm(\pybemat')$ are fully defined on all $c'\leq c$, and  
  	\item all $\leqC$-minimal extensions $\ybemat\in\xcycsets{\pybemat}$ are also extensions of $\pybemat'$.
  \end{enumerate}
 \end{itemize}
\end{theorem}
Hence, when given a partial cycle set $\pybemat$, not only permutations $\perm$ for which $\perm(\pybemat)\lneqP\pybemat$ are useful, but also those for which $\perm(\pybemat)\leqP_c\pybemat$ (if $\pybemat$ is not complete yet).

\subsection{Clause Learning} 
In this section, we explain which clauses are constructed after the minimality check has found a permutation $\perm$ that can be used to refine or exclude the current assignment $\asg$.
Let us first focus on the case where $\perm$ is a witness of non-minimality, i.e., we will describe the clause that justifies backtracking. 

In principle, we could learn a clause that simply eliminates the current partial assignment by $\bigvee_{\ell\in\asg}\olnot \ell$, i.e., stating that at least one literal needs to be different. However, we want to learn a clause that captures the essence of why $\perm$ is a witness of minimality so that it excludes as many assignments as possible. 
Explaining why $\perm$ is a witness of non-minimality boils down to:
 \begin{itemize}
 	\item explaining why $\perm(P)_{c'}\leqP P_{c'}$ for each ${c'}< c$, and 
 	\item explaining why $\perm(P)_{c}\lneqP P_{c}$.
 \end{itemize}
 Hence, the clause will state that (at least) one of these conditions has to change. 
For  $\perm(P)_{c'}\leqP P_{c'}$ to change, either $\perm(P)_{c'}$ should be able to take a value larger than $\min P_{c'}$ or $P_{c'}$ should be able to take a value below $\max \perm(P)_{c'}$. The explanation for $\perm(P)_{c}\lneqP P_{c}$ is similar. 
As such, the clause that we learn is: 
\begin{equation}\label{eq:clause}
	\bigvee_{x\geq \min P_{c}}\matvar_{\pi(c),\pi(x)} \vee 
	\bigvee_{x\leq \max \pi(P)_{c}} \matvar_{c,x} \vee 
	\bigvee_{c'<c}\left(
		\bigvee_{x> \min P_{c'}}\matvar_{\pi(c'),\pi(x)} \vee 
		\bigvee_{x< \max \pi(P)_{c'}} \matvar_{c',x}
		\right)
		\end{equation}
		
\subsubsection{Propagation}
For propagation, everything is completely analogous to the case where we find a witness of non-minimality.
The reason is that literal $\ell$ can be propagated in $\asg$ using Theorem \ref{thm:main} if and only if $\perm$ would be a witness of non-minimality in $\asg\cup\{\olnot\ell\}$. 
There is one important design decision to make here; often there are multiple literals that can be propagated. 
We can add a clause for every such literal, or we can choose one of them and only propagate that. 
The propagation of a single literal often causes more propagations at the same row of the cycle set, resulting in further eliminations of values from the same cell.
If our minimality check detects multiple potential propagations, they will always be over the same cell. 
In this case, one could prefer to avoid adding propagation constraints for things that can already be caught by the unit propagation phase that follows.
So, while adding all clauses is feasible, we have chosen to only propagate one literal at a time and let the solver continue. 

\subsubsection{Optimization} 
The clause that is learned in this way can be optimized further by taking the problem structure into account in some ways. In particular, we know that for several sets of variables an \texttt{ExactlyOne} constraint holds (\eg all variables $\matvar_{c,\cdot}$ for any cell $c$, or all variables $\matvar_{i,\cdot,x}$ for a fixed $i$ and $x$). 
If for such a set $S$ of variables (for which we know exactly one will be true in every satisfying assignment) and for some variable $s\in S$, our clause $C$ contains $\bigvee_{s'\in S,s'\neq s} s'$, then this can be replaced by simply stating $\olnot s$. 
Indeed, since exactly one of the variables holds, stating that the last one is false is equivalent to stating that one of the others is true.
Similarly, if the clause already contains $\olnot s$, all literals $s'$ with $s'\in S$ can safely be removed from the clause. 

For example, assume we are searching for cycle sets of size $4$ and have fixed a diagonal that has  $\pybemat_{\tuple{1,1}}=1$. Suppose the clause \eqref{eq:clause} is $\matvar_{1,2,2}\lor\matvar_{1,2,3}$, then this expresses that the cell $\tuple{1,2}$ should contain a $2$ or a $3$. Now this is the same as saying that it cannot contain the value four: $\olnot\matvar_{1,2,4}$. 
This simple trick can be used to get shorter clauses with better propagation properties. 

\section{Implementating the Minimality Check}
Given a partial cycle set $\pybemat\in\pcycsets$, the goal of the minimality check is to find an isomorphism $\perm$ of the problem that guarantees that  $\perm(\pybemat)\leqP(\pybemat)$ so that \cref{thm:main} guarantees that we can either discard this partial cycle set or refine it. 
Inspired by the original SMS papers, the minimality check can perform a backtracking search (over the space of permutations!) with a strong pruning mechanism. 
It iteratively refines a so-called \emph{partial permutation}, until either a useful permutation (satisfying one of the conditions of Theorem \ref{thm:main}) is found, or until it is certain that no such permutations exist.\footnote{We will discuss that it is sometimes useful not to do a complete check here, as long as completeness is guaranteed on complete assignments.}

On the other hand, the minimality check can also be viewed as a combinatorial search problem.
More specifically, given the current assumptions, \ie the current (partial) cycle set, we want to decide whether there exists an isomorphism that is a witness of non-minimality.
If no such permutation exists, we know that the current cycle set is lexicographically minimal; otherwise, we can use the satisfying assignment to extract the witness.

\subsection{Backtracking Approach}
We first describe the backtracking approach based on the original SMS paper~\mycite{SMS}.
A \emph{partial permutation} is a function $\perm:\set\to2^\set\setminus\emptyset$ that maps every element of $\set$ to a non-empty set of values.
Intuitively, the value $\perm(x)$ represents the possible images of $x$ under any completion of $\perm$.
A partial permutation is called \emph{complete} if for all $x\in\set$, $\perm(x)$ is a singleton and if $\perm(x)\neq\perm(y)$ for all $y\in\set$ where $y\neq x$.
In this case, we identify $\perm$ with the actual permutation.
Given two partial permutations $\perm$ and $\perm'$, we say that $\perm'$ \emph{extends} $\perm$ if and only if for all $x\in\set$ it holds that $\perm'(x)\subseteq\perm(x)$.

Our algorithm will recursively refine the permutation using the given partial cycle set until it holds that for all complete extensions $\perm'$ of $\perm$, $\perm'(\pybemat)\leqP\pybemat$ (where the inequality will be guaranteed to be strict if $\pybemat$ is complete).
To do this, we recursively assign a value $y$ to $\perm(x)$, starting with $x=1$, and refine the obtained partial permutation using available (domain-specific) knowledge.
To refine a partial permutation $\perm$ after making a new assignment, we need to take the following properties into account:
\begin{itemize}
 \item the image of $\pybemat$ under $\perm$ should be smaller than or equal to $\pybemat$, and
 \item the partial permutation $\perm$ can be extended to an isomorphism of the problem.
\end{itemize}
As mentioned earlier, the algorithm first assigns a value $x_1\in\perm(1)$ to $\perm'(1)$.
Because the completion of $\perm'$ needs to be well-defined, we refine $\perm'$ such that $x_1$ is no longer an option for other elements in $\set\setminus\{1\}$.
To ensure that $\perm'(\pybemat)\leqP\pybemat$, we need to enforce that $\perm'(\pybemat)_{\tuple{1,1}}\leqP\pybemat_{\tuple{1,1}}$, or equivalently, $\max{\perm'(\pybemat)_{\tuple{1,1}}}\leq\min{\pybemat_{\tuple{1,1}}}$.
Hence, we can refine $\perm'$ such that \[\max{\perm'(\pybemat)_{\tuple{1,1}}} = \max\perm'^{-1}(\pybemat_{\tuple{\perm'(1),\perm'(1)}})\leq\min{\pybemat_{\tuple{1,1}}}.\] Some care is needed here: even though we know the possible values in the cell $\pybemat_{\tuple{\perm'(1),\perm'(1)}}$, since we are still constructing $\pi'$, we do not know exactly what the values in $\perm'^{-1}(\pybemat_{\tuple{\perm'(1),\perm'(1)}})$ will be. This will on the one hand cause extra propagation: if $x\in \pybemat_{\tuple{\perm'(1),\perm'(1)}}$, we will propagate that $\pi^{-1}(x)\neq y$ for any $y>\min P_{\tuple{1,1}}$. But after doing this there might still be multiple options left for this cell. In that case, we make a further choice on refining $\perm$ so that for each $x\in \pybemat_{\tuple{\perm'(1),\perm'(1)}}$, $\perm^{-1}(x)$ is fixed (and again, bijectivity is enforced).
	 	
Last, we need to ensure that the extended permutation $\perm'$ can still be extended to an isomorphism of the problem.
If the diagonal is not fixed, all permutations $\perm\in\symmgroup_n$ are isomorphisms and hence no extra measures need to be taken.
If however, a diagonal $T$ is fixed, only the permutations $\perm\in\centr_{\symmgroup_n}(T)$ are isomorphisms of the current problem being solved. 
Note that a permutation $\perm$ fixes the diagonal if it maps each cycle of the diagonal onto a same-length cycle, in the same order.
As a consequence, as soon as $\perm'(x)$ is fixed, for each $y$ in the same cycle as $x$, $\perm'(y)$ will be fixed accordingly. For example, consider the diagonal $(2 3 1)(5 6 4)$. In case we fix $\perm'(1)=6$, we will need to map $\perm'(3)=5$ and $\perm'(2)=4$ in order to fix this diagonal.

As soon as this is done, if we obtained  $\perm'(\pybemat)_{\tuple{1,1}}\lneqP\pybemat_{\tuple{1,1}}$, we have found a witness of non-minimality. 
Otherwise, if $\perm'(\pybemat)_{\tuple{1,1}}\leqP\pybemat_{\tuple{1,1}}$ and either $\perm'(\pybemat)_{\tuple{1,1}}$ or $\pybemat_{\tuple{1,1}}$ is not complete (i.e., still allows multiple values), we can choose to either learn clauses that eliminate all but one remaining values in the cell, or continue the search for a witness of non-minimality. 
When continuing the search, we repeat the same process for the next cell $\tuple{1,2}$, etcetera. 
As soon as $\perm$ is fully defined, it can be verified without further choices against the remaining cells of $\pybemat$. 
Whenever this process gets stuck,  we backtrack and make a different assignment to $\perm(x)$ for the last assigned $x$.

For the checks of partial cycle sets, the search can be aborted early, as a result, the check is postponed and repeated when more information is available.
Specifically, users can limit the number of nodes visited in the search tree (in a depth-first fashion).
Next, users can also limit the number of partial cycle sets that are checked by setting a frequency at which the minimality check should be executed.
For complete cycle sets, a full search is necessary to guarantee complete symmetry breaking.

\subsubsection{Representing Partial Permutations}
We give a short description of the internal representation of partial permutations.
The internal representation used in this tool was inspired by the one used in SMS.
We represent partial permutations as an ordered partition of $\set$, i.e., $\perm=[\set_1,\set_2,\ldots,\set_m]$, where $\set_1\cup\ldots\cup\set_m=\set$ and $\set_1,\set_2,\ldots,\set_m$ are pairwise disjoint. 
The ordered partition represents the partial permutation $\perm$ where $\perm^{-1}(x_1)<\perm^{-1}(x_2)$ for all $x_1\in\set_i,x_2\in\set_j$ with $i<j$. 
In other words, given $x_1<x_2$, $\perm$ maps $x_2$ to the same set as $x_1$ or to a set that appears after that set in the ordered partition.
Given the ordered partition $\perm=[\{6,5,4\},\{3\},\{2, 1\}]$, we have that $\perm(1)=\perm(2)=\perm(3)=\{4,5,6\}$, $\perm(4)=3$ and $\perm(5)=\perm(6)=\{1,2\}$.

Note that extracting a fully defined permutation $\perm'$ from an ordered partition $\perm$ can be done quite easily. 
We represent the ordered partition $\perm$ using two vectors; a vector of integers $V$, representing the content of the partitions and a vector which signals the start of a partition, $P$. 
In this case, the vector $V$ already represents a valid permutation $\perm'$ if we define $\perm'(x)=y$ if and only if $V[x]=y$.

This representation becomes even more useful when considering fixed diagonals.
In this case, we only need to consider complete permutations $\perm$ that leave the diagonal invariant.
Hence, we can use the initial ordered partition to represent only those permutations.

\subsection{Incremental SAT-Based Approach}
The minimality check itself can also be viewed as a combinatorial search problem.
In this paper, we have chosen to encode the minimality check using a propositional formula, whose satisfiability is determined using a second SAT-solver.
Specifically, we encode what it means for an isomorphism of the problem to be a witness of non-minimality.
In order to reason about a specific (partial) cycle set, we can use so-called \emph{assumptions} to fix the truth values for variables representing the cycle set.
As such, we can use the same SAT-solver to reason about many same-sized cycle sets, while keeping track of previously learned clauses.
In the future, this incremental minimality check can be extended to also detect whether the current partial cycle set could be refined.
This technique is similar to the co-certificate learning added to SAT Modulo Symmetries~\cite{KPS23Co-CertificateLearningSATModuloSymmetries}, however co-certificate learning is concerned with properties that are orthogonal to the minimality check.
To the best of our knowledge, the minimality check itself has never been performed using a SAT oracle.

\newcommand\omatvar{w\xspace}
\newcommand\nbvar{n\xspace}
\newcommand\pvar{p\xspace}
\newcommand\mat{M\xspace}
\subsubsection{Minimality Check Encoding}
\label{ss:MCE}
Before describing how this incremental minimality check works, we first give a high-level overview of the used encoding, a more detailed description can be found in Appendix \ref{app:mincheckCNF}.
The problem that we encode in propositional logic verifies whether there exists a witness of non-minimality (i.e., a permutation $\pi$ for which \cref{def:lneqP} holds) for a given cycle set $(\set,\binop)$ represented by matrix $\mat$.
We start by describing an encoding for the minimality check of complete cycle sets before adapting it in order to also deal with partial cycle sets.

First, we introduce Boolean variables that represent the matrix $\mat$ and variables that represent its image $\pi(M)$, also referred to as $M'$ here.
These matrices are encoded using the same one-hot encoding described in \cref{ss:satenc}.
The first cycle set, represented by matrix $\mat$, is encoded using the variables $\omatvar_{i,j,k}$ (for all $i,j,k\in\set$), where $\omatvar_{i,j,k}$ is true if $\mat_{i,j}=k$.
Similarly, we encode its image represented by the matrix $\mat'$, using the variables $\omatvar'_{i,j,k}$ for all $i,j,k\in\set$.
In order to represent the permutation $\pi$, we introduce the variables $\pvar_{i,j}$ for all $i,j\in\set$ where $\pvar_{i,j}$ is true if $\pi(i)=j$.

Next, we add clauses to ensure that the permutation is indeed a well-defined isomorphism of the problem and that $\mat'$ equals $\pi(M)$.
Furthermore, we ensure that $\pi$ is a witness of non-minimality (i.e., that $M'=\pi(M)<M$). 
To do this, we add static symmetry breaking constraints inspired by the compact encoding for lex-leader constraints~\cite{DBBD16ImprovedStaticSymmetryBreakingSAT}.
First, we introduce an order $\leq_{\omatvar'}$ over the variables used to encode the matrix $\mat'$.
This order should ensure that an assignment is lexicographically minimal if $\mat'=\pi(\mat)\leq\mat$, assuming that $\lfalse<\ltrue$.
We encode the static symmetry breaking constraints using the auxiliary variables $n_{c,i}$ which are true if $\omatvar'_{c',i'}\Leftrightarrow\omatvar_{c',i'}$ for all $\omatvar'_{c',i'}\leq_{\omatvar'}\omatvar'_{c,i}$.

Because we are only considering complete cycle sets we can add redundant \texttt{ExactlyOne} constraints over the matrix variables (similar to those added in the encoding of the original problem).
These constraints ensure that each matrix cell is assigned exactly one value, and that rows in the matrices contain unique values.
For the original matrix $\mat$, this will not have any impact since all variables are fixed using assumptions, however, in context of the image of $\mat$, \ie $\mat'$, these constraints will enhance propagation.

In order to reason about partially defined cycle sets, some minor adaptations to this encoding are needed.
First, the variables representing the matrices get a slightly different meaning; when $\omatvar^{(')}_{c,k}$ is true it no longer holds that $\mat^{(')}_{c}=k$, but instead we have that $k\in\mat^{(')}_{c}$.
As a result, it is no longer correct to include the redundant \texttt{ExactlyOne} constraints over these variables. 
Next, to decide whether $\pi$ is a witness of non-minimality, we now use the order defined in \cref{def:lneqP}, and ensure that $M\lneqP M'$.
To do this, we add variables $l_{c,k}$ and $g_{c,k}$ for all cells $c$ and $k\in\set$, which are true respectively if $M'_{c}<k$ and $M_{c}>k$.
These variables allow us to reason about the maximum and minimum possible values of a matrix cell.
Next, we introduce the variables $\nbvar'_{c,k}$ for all cells $c$ and $k\in\set$, which similar to the $\nbvar_{c,k}$-variables indicate that $\perm$ can still be a witness of non-minimality at this point.
Similar to the encoding described in \cref{ss:diagfix}, both the partial and complete encoding can be optimized when a diagonal is fixed.

Note that the partial encoding can also be used to perform the minimality check for complete cycle sets.
However, preliminary experiments showed that using separate complete and partial minimality checks results in a faster enumeration time, hence this is the setup used in our experiments below.

\subsubsection{The Incremental Minimality Check}
During the search, the enumerating SAT solver will make calls to the minimality check to verify whether the current (partial) assignment is lexicographically minimal.
The incremental minimality check consists of one (or two, if a separate encoding is used for complete cycle sets) incremental SAT solvers that repeatedly try to find satisfying assignments for the given formula.
To ensure that the solvers reason about the current cycle set, we use assumptions to ensure that the $\omatvar$-variables defined in the previous section get the correct truth values.
If the SAT solver can find a satisfying assignment, there exists a witness of non-minimality for the assumed cycle set.
The minimality check will extract this permutation from the satisfying assignment to create a breaking clause that excludes the current assignment (and all of its extensions).
This clause can then in turn be learned by the enumerating SAT solver.
If, on the other hand, the formula is unsatisfiable given the current assumptions, the current solution is minimal.
Note that if the current assignment is complete, this implies that we have found a lexicographically minimal solution which can be added to the list of non-isomorphic solutions.

Similar to the backtracking approach, one can opt to not run partial minimality checks until completion.
In this context, one could decide to limit the number of conflicts or decisions or to only add a selection of the symmetry breaking constraints.
Once again, the frequency of the minimality check for partial cycle sets can also be adapted.
For complete cycle sets, no such limits can be imposed, in order to guarantee a fully non-isomorphic enumeration.

\section{Experimental Evaluation}
\label{sec:experiments}
Our enumeration tool, \texttt{YBE-SMS}, is implemented on top of the state-of-the-art SAT-solver \texttt{CaDiCal}~\cite{BiereFazekasFleuryHeisinger-SAT-Competition-2020-solvers} (version 1.9.4), extended with a so-called \emph{user-propagator} implemented using the \texttt{IPASIR-UP}-API~\cite{FNPKSB23IPASIR-UPUserPropagatorsCDCL}\footnote{The tool is available on \url{https://gitlab.kuleuven.be/krr/software/ybe-sms}.}.
This API allows users to closely interact with the solver, for example by tracking variable assignments, choosing variables to branch on and most importantly by adding new clauses to the solver \emph{during} the search process, which allows users to implement custom propagators.
In our case, the custom propagator will ensure that current assignments remain lexicographically minimal by performing the minimality check.
If the minimality check fails, a clause is constructed in order to make the solver backtrack or propagate.
The minimality check can either be performed using the backtracking approach or by using a second (incremental) instance of the \texttt{CaDiCal} solver.
Each time when the minimality is performed, the current partial assignment is passed to that second instance by means of assumptions. 
All experiments were performed on a machine with an AMD(R) Genoa-X CPU running Rocky Linux 8.9 with Linux kernel 4.18.0.

\subsection{Generating Non-Degenerate Involutive Solutions of Size 10}
In the first set of experiments, we compare the new \texttt{YBE-SMS} implementations against the implementation of Akg\"un, Mereb and Vendramin~\cite{AMV22Enumerationset-theoreticsolutionsYang-Baxterequation} (referred to as AMV22).
To do this, we have identified configurations that find a suitable balance between the frequency of performing this check and the limitations we impose on it; these details can be found in Appendix \ref{app:conf}.
In \cref{fig:gapvssms}, we see that both new approaches are significantly faster than AMV22 for all sizes of the problem up to 10 (which is the highest cardinality solved by AMV22). 
For example, enumerating all solutions of size 10 using the AMV approach took more than 8 days.
The same enumeration with the backtracking approach took only 12 hours and with the incremental approach all solutions of size 10 were enumerated in less than two hours.
For the backtracking approach, this speed-up seems to diminish as the problem grows but for the incremental approach, the opposite appears to be true.
This new incremental approach is approximately 106 times faster than the AMV22 approach at enumerating all solutions of size 10, the backtracking approach only reaches a speedup of 16.1 for the enumeration of size 10 compared to AMV22.
As a result, with this new incremental approach, we were able to enumerate all solutions of size 11 which is a novel result, of value to people studying the YBE.
A more detailed view of the results in \cref{fig:gapvssms} is given in \cref{tab:gapvsssms} which can be found in Appendix \ref{app:expr}.

\begin{figure}[]
	\centering
	\includegraphics[width=0.7\linewidth]{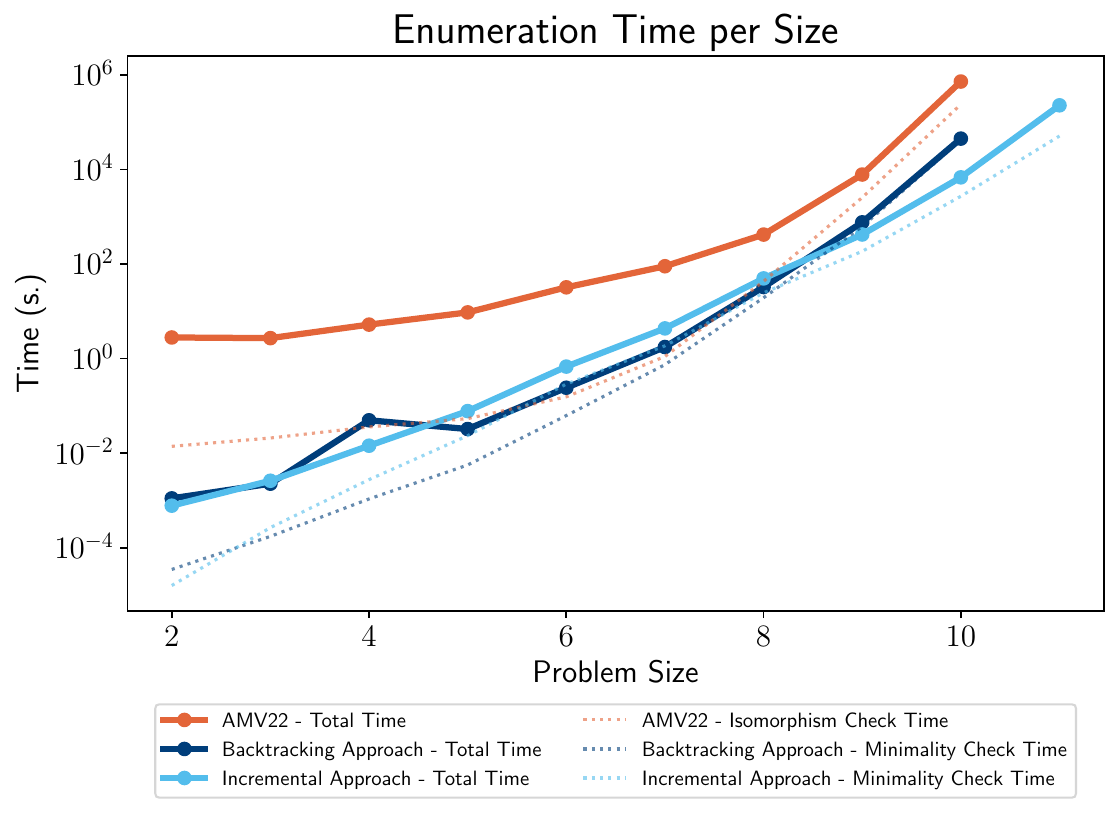}
	\caption{Comparing the runtimes of the implementation of AMV22 and our approaches building on SAT Modulo Symmetries.}
	\label{fig:gapvssms}
\end{figure}

\subsection{Comparing the SMS Approaches}
Previously, no database containing the solutions for the YBE over a set $\set$ with cardinality $|\set|=11$ was known.
Using the backtracking approach discussed here, we were able to enumerate all solutions for all diagonals of this problem except the diagonal that equals the identity.
However, with the incremental approach, we were able to improve on these results, and as such we have constructed a complete database of solutions over a set $\set$ with cardinality $|\set|=11$.

In order to understand why the backtracking approach was not able to enumerate all solutions, we study the time usage of the minimality check in more detail.
We break down the time usage of the minimality check into four categories:
\begin{itemize}
	\item checking partial cycle sets where a clause was added (i.e., a witness or refinement was found),
	\item checking partial cycle sets where no clause was added,
	\item checking fully defined cycle sets where a clause was added (i.e., a witness was found) and,
	\item checking fully defined cycle sets where no clause was added (i.e., the cycle set is lexicographically
	minimal).
\end{itemize}
In \cref{tab:mcperdiag_bt}, we observe that as size increases, the backtracking approach spends a larger amount of time verifying whether lexicographically minimal cycle sets are indeed minimal.
To understand why this is the case, note that in these cases all possible symmetries of the problem need to be considered.
In the very worst case, when enumerating cycle sets of size $n$ with the diagonal fixed to the identity, where all rows equal the identity as well; all $n!$ permutations need to be considered.
This explains why the speed-up diminishes as the problem grows and why we were unable to expand the results beyond size 10.
In \cref{tab:mcperdiag_incr}, we observe that the incremental approach does not have these issues and as a result, it was able to enumerate all solutions for size 11.
Note that with this approach, the identity diagonal is no longer the bottleneck.
Instead, for size 11, we see that enumerating the diagonal $(12)(34)$ took longer than enumerating the identity diagonal.
There are several reasons why this might be the case, one possible reason is that there are so many solutions to exclude (i.e., $\pm 11\,800\,000$ compared to $\pm 1\,000\,000$ original clauses) that the SAT-solver is severely slowed down by the solution-excluding constraints.
We hope that further experiments will give us more insights into why this is the case. 

\begin{table}[t]
	\centering
	\resizebox{\textwidth}{!}{
	\begin{tabular}{lll|lllllll|l|}
	\cline{4-11}
											  &             &              		& \multicolumn{8}{c|}{\textbf{Backtracking Approach}}                                                                                                                                                                                                                                                                                                                                                                                                                                                                   \\ \cline{2-11} 
	\multicolumn{1}{l|}{}                     & Diag.       & \# Sols.     		& \multicolumn{1}{l|}{\begin{tabular}[c]{@{}l@{}}Time\\ (\% of total time)\end{tabular}} & \begin{tabular}[c]{@{}l@{}}Number of \\ partial checks\end{tabular}				& \begin{tabular}[c]{@{}l@{}}Partial check,\\ Non-minimal \\ (\% of Time)\end{tabular} & \multicolumn{1}{l|}{\begin{tabular}[c]{@{}l@{}}Partial check,\\ No conclusion\\ (\% of Time)\end{tabular}} & \begin{tabular}[c]{@{}l@{}}Number of \\complete checks\end{tabular}	& \begin{tabular}[c]{@{}l@{}}Complete check,\\ Non-minimal\\ (\% of Time)\end{tabular} 	& \multicolumn{1}{l|}{\begin{tabular}[c]{@{}l@{}}Complete check,\\ Minimal\\ (\% of Time)\end{tabular}} 	& \begin{tabular}[c]{@{}l@{}}Mininimality Check\\ (\% of Time)\end{tabular} 	\\ \hline
	\multicolumn{1}{|l|}{\multirow{3}{*}{8}}  & \textbf{id} & 2\,041         	& \multicolumn{1}{l|}{\textbf{15.45s. (47.33)}}                                          &	6147																								& 0.78                                                                         		& \multicolumn{1}{l|}{7.20}                                                                           		&  2815																	& 5.66                                                                         			& \multicolumn{1}{l|}{76.25}                                                               					& \textbf{89.90}                                 								                         			\\
\multicolumn{1}{|l|}{}                    & (12)        & 4\,988         	& \multicolumn{1}{l|}{2.87s. (8.78)}                                                     	 &	2953																								& 1.65                                                                         		& \multicolumn{1}{l|}{9.26}                                                                           			&  5322  																& 1.51                                                                       			& \multicolumn{1}{l|}{45.39}                                                                        		& 57.81                                 									                               	\\
\multicolumn{1}{|l|}{}                    & (12)(34)    & 7\,030         	& \multicolumn{1}{l|}{2.48s. (7.59)}                                                     	 &	3538																								& 1.76                                                                         		& \multicolumn{1}{l|}{9.65}                                                                           			&  7375  																	& 0.93                                                                       			& \multicolumn{1}{l|}{27.82}                                                                        		& 40.17                                 									                              	\\ \hline
\multicolumn{1}{|l|}{\multirow{3}{*}{9}}  & \textbf{id} & 15\,534        	& \multicolumn{1}{l|}{\textbf{514.63s. (67.67)}}                                         	 &	43957																								& 0.09                                                                         		& \multicolumn{1}{l|}{1.87}                                                                           			&  18413 																	& 2.78                                                                       			& \multicolumn{1}{l|}{92.41}                                                              					& \textbf{97.15}                                  							                    				\\
\multicolumn{1}{|l|}{}                    & (12)        & 41\,732        	& \multicolumn{1}{l|}{68.82s. (9.05)}                                                    	 &	19328																								& 0.35                                                                         		& \multicolumn{1}{l|}{4.54}                                                                           			&  43008 																	& 1.13                                                                       			& \multicolumn{1}{l|}{75.90}                                                                        		& 81.92                                 									                             		\\
\multicolumn{1}{|l|}{}                    & (12)(34)    & 61\,438        	& \multicolumn{1}{l|}{37.69s. (4.96)}                                                    	 &	23178																								& 0.58                                                                         		& \multicolumn{1}{l|}{9.62}                                                                           			&  62538 																	& 0.46                                                                       			& \multicolumn{1}{l|}{46.78}                                                                        		& 57.44                                 									                            		\\ \hline
\multicolumn{1}{|l|}{\multirow{3}{*}{10}} & \textbf{id} & 150\,957       	& \multicolumn{1}{l|}{\textbf{35\,396.79s. (79.02)}}                                     	 &	371265																								& 0.01                                                                         		& \multicolumn{1}{l|}{0.28}                                                                           			&  163126																	& 1.13                                                                       			& \multicolumn{1}{l|}{98.03}                                                               					& \textbf{99.45}                                 							                   					\\
\multicolumn{1}{|l|}{}                    & (12)        & 474\,153       	& \multicolumn{1}{l|}{3\,998.35s. (8.93)}                                                	 &	184837																								& 0.04                                                                         		& \multicolumn{1}{l|}{1.18}                                                                           			&  480790																	& 0.58                                                                       			& \multicolumn{1}{l|}{92.08}                                                                        		& 93.88                                 									                           			\\
\multicolumn{1}{|l|}{}                    & (12)(34)    & 807\,084       	& \multicolumn{1}{l|}{1\,380.82s. (3.08)}                                                	 &	251264																								& 0.12                                                                         		& \multicolumn{1}{l|}{4.57}                                                                           			&  814658																	& 0.31                                                                       			& \multicolumn{1}{l|}{63.98}                                                                        		& 68.97                                 									                          			\\ \hline
\multicolumn{1}{|l|}{\multirow{3}{*}{11}} & \textbf{id} & 1\,876\,002  		& \multicolumn{1}{l|}{}                                                                  	 &																									&                                                                              			& \multicolumn{1}{l|}{    }                                                                           			&        																	&                                                                            			& \multicolumn{1}{l|}{}                                                                             		&                                       									                   			\\
\multicolumn{1}{|l|}{}                    & (12)        & 6\,563\,873  		& \multicolumn{1}{l|}{364\,026.91s.}                                                       	 &	2175972																								& 0.00                                                                         		& \multicolumn{1}{l|}{0.21}                                                                          			&  6593100      																	& 0.19                                                                       			& \multicolumn{1}{l|}{96.61}                                                                        		& 97.02                                 									                       		\\
\multicolumn{1}{|l|}{}                    & (12)(34)    & 11\,807\,217 		& \multicolumn{1}{l|}{102\,863.75s.}                                              		 	 &	3131880																								& 0.01                                                                         		& \multicolumn{1}{l|}{1.27}                                                                           			&  11839127      																	& 0.11                                                                       			& \multicolumn{1}{l|}{67.30}                                                                        		& 68.72                                 									                 				\\ \hline \\
	\end{tabular}
	}
	\caption{A breakdown of the time-usage of the backtracking approach for the diagonals $id, (12)$ and $(12)(34)$ of sizes 8,9,10 and 11. 
	Note that we were unable to enumerate the identity diagonal for size 11 using this approach.} 
	\label{tab:mcperdiag_bt}
	\end{table}

\begin{table}[t]
\centering
\resizebox{\textwidth}{!}{
\begin{tabular}{lll|llllllll|}
\cline{4-11}
                                          &             &              		& \multicolumn{8}{c|}{\textbf{Incremental Approach}}                                                                                                                                                                                                                                                                                                                                                                                                                                                                    \\ \cline{2-11} 
\multicolumn{1}{l|}{}                     & Diag.       & \# Sols.     		& \multicolumn{1}{l|}{\begin{tabular}[c]{@{}l@{}}Time\\ (\% of total time)\end{tabular}} & \begin{tabular}[c]{@{}l@{}}Number of\\partial checks\end{tabular}	& \begin{tabular}[c]{@{}l@{}}Partial check,\\ Non-minimal \\ (\% of Time)\end{tabular} 			& \multicolumn{1}{l|}{\begin{tabular}[c]{@{}l@{}}Partial check,\\ No conclusion\\ (\% of Time)\end{tabular}} 	&\begin{tabular}[c]{@{}l@{}}Number of\\ complete checks\end{tabular}											& \begin{tabular}[c]{@{}l@{}}Complete check,\\ Non-minimal\\ (\% of Time)\end{tabular} 	& \multicolumn{1}{l|}{\begin{tabular}[c]{@{}l@{}}Complete check,\\ Minimal\\ (\% of Time)\end{tabular}} & \begin{tabular}[c]{@{}l@{}}Mininimality Check\\ (\% of Time)\end{tabular} \\ \hline
\multicolumn{1}{|l|}{\multirow{3}{*}{8}}  & \textbf{id} & 2\,041         	& \multicolumn{1}{l|}{\textbf{14.53s. (29.27)}}                                          &	1097																					& 0.58                                                                      & \multicolumn{1}{l|}{2.67}                                                                           			&	4582							& 15.05                                                                        			& \multicolumn{1}{l|}{40.44}                                                               				& \textbf{58.73}                                                            \\
\multicolumn{1}{|l|}{}                    & (12)        & 4\,988         	& \multicolumn{1}{l|}{4.08s. (8.21)}                                                     &	221																					& 0.27                                                                         	& \multicolumn{1}{l|}{2.03}                                                                           			&	6518							& 11.70                                                                        			& \multicolumn{1}{l|}{43.23}                                                                        	& 57.22                                                                     \\
\multicolumn{1}{|l|}{}                    & (12)(34)    & 7\,030         	& \multicolumn{1}{l|}{4.52s. (9.10)}                                                     &	250																					& 0.35                                                                       	& \multicolumn{1}{l|}{1.54}                                                                           		&	8596							& 9.17                                                                        			& \multicolumn{1}{l|}{43.04}                                                                        	& 54.09                                                                     \\ \hline
\multicolumn{1}{|l|}{\multirow{3}{*}{9}}  & \textbf{id} & 15\,534        	& \multicolumn{1}{l|}{\textbf{135.86s. (32.25)}}                                         &	7898																					& 0.53                                                                      & \multicolumn{1}{l|}{4.01}                                                                           		&	25164							& 10.95                                                                      			& \multicolumn{1}{l|}{42.71}                                                               				& \textbf{58.19}                                                            \\
\multicolumn{1}{|l|}{}                    & (12)        & 41\,732        	& \multicolumn{1}{l|}{37.68s. (8.95)}                                                    &	1639																					& 0.17                                                                      & \multicolumn{1}{l|}{2.17}                                                                           			&	48124							& 8.21                                                                      		   	& \multicolumn{1}{l|}{41.75}                                                                        	& 52.31                                                                     \\
\multicolumn{1}{|l|}{}                    & (12)(34)    & 61\,438        	& \multicolumn{1}{l|}{42.00s. (9.97)}                                                    &	1612																					& 0.20                                                                      & \multicolumn{1}{l|}{1.84}                                                                           		&	67463							& 5.79                                                                         			& \multicolumn{1}{l|}{40.22}                                                                        	& 48.05                                                                     \\ \hline
\multicolumn{1}{|l|}{\multirow{3}{*}{10}} & \textbf{id} & 150\,957       	& \multicolumn{1}{l|}{\textbf{1\,073.65s. (15.80)}}                                      &	47416																					& 0.34                                                                      & \multicolumn{1}{l|}{5.15}                                                                           		&	189756							& 7.92                                                                         			& \multicolumn{1}{l|}{50.25}                                                               				& \textbf{63.66}                                                            \\
\multicolumn{1}{|l|}{}                    & (12)        & 474\,153       	& \multicolumn{1}{l|}{605.32s. (8.91)}                                                   &	15056																					& 0.10                                                                      & \multicolumn{1}{l|}{2.78}                                                                           		&	504172							& 4.03                                                                         			& \multicolumn{1}{l|}{40.61}                                                                        	& 47.52                                                                     \\
\multicolumn{1}{|l|}{}                    & (12)(34)    & 807\,084       	& \multicolumn{1}{l|}{817.65s. (12.03)}                                                  &	15174																					& 0.06                                                                      & \multicolumn{1}{l|}{1.86}                                                                           			&	837097							& 2.39                                                                         			& \multicolumn{1}{l|}{36.38}                                                                        	& 40.70                                                                     \\ \hline
\multicolumn{1}{|l|}{\multirow{3}{*}{11}} & \textbf{id} & 1\,876\,002  		& \multicolumn{1}{l|}{16\,153.55s. (7.14)}                                               &	435489																					& 0.23                                                                      & \multicolumn{1}{l|}{5.10}                                                                           			& 2041121								& 4.04                                                                        			& \multicolumn{1}{l|}{43.28}                                                                        	& \textbf{52.65}                                                                     \\
\multicolumn{1}{|l|}{}                    & (12)        & 6\,563\,873  		& \multicolumn{1}{l|}{17\,578.69s. (7.76)}                                               &	137194																					& 0.03                                                                      & \multicolumn{1}{l|}{1.31}                                                                           			& 6692410								& 1.03                                                                         			& \multicolumn{1}{l|}{26.40}                                                                        	& 28.77                                                                     \\
\multicolumn{1}{|l|}{}                    & (12)(34)    & 11\,807\,217 		& \multicolumn{1}{l|}{\textbf{39\,061.15s. (17.25)}}                                     &	170158																					& 0.01                                                                      & \multicolumn{1}{l|}{0.71}                                                                           			& 11917937								& 0.33                                                                         			& \multicolumn{1}{l|}{17.71}                                                                        	& 18.76                                                                     \\ \hline \\
\end{tabular}
}
\caption{A breakdown of the time-usage of the incremental approach for the diagonals $id, (12)$ and $(12)(34)$ of sizes 8,9,10 and 11.}
\label{tab:mcperdiag_incr}
\end{table}

\section{Conclusions and Future Work}
\label{sec:concl}
In this paper, we showed how to apply the SAT Modulo Symmetries framework to the generation of non-degenerate, involutive, set-theoretic solutions to the Yang-Baxter equation. 
Our methods outperform the state-of-the-art on large instances by two orders of magnitude, and we have extended known results to include non-degenerate cycle sets of size 11.

Databases of solutions up to size 10 have inspired a vast body of research in math research.
For instance, a recent survey \cite{zbMATH07919707} contains several mathematical conjectures that were inspired by mining the database of solutions up to size 10 for interesting properties (see for instance Problems 57 and 61).
The list of solutions for size 11 will be useful for similar purposes and provides a substantial number of decomposable or multipermutation solutions (both of which are very important in the combinatorial theory of the Yang-Baxter equation). 
This increased data gives us the opportunity to better understand how solutions can be constructed from smaller components.

One important question that might remain is why one should trust our implementation, except for the fact that up to size 10 our results coincide with what is known so far. 
In combinatorial optimization, \emph{proof logging}, which is the idea that solvers should not just output a solution (or in this case, a set of solutions), but also a \emph{machine-checkable proof} that their answer is indeed correct is gaining popularity. 
The SAT solver \texttt{CaDiCaL} that underlies SMS supports proof logging.
However, this is a very weak guarantee as it only allows checking that the SAT solver did not make any mistakes, and does not provide any guarantees whatsoever on the encoding or on the correctness of the custom propagator. 
The most promising approach to achieve proof logging for SMS appears to be using the VeriPB proof system, which was recently used to certify static symmetry breaking \cite{BGMN23CertifiedDominanceSymmetryBreakingCombinatorialOptimisation}. However, there are many challenges on the road ahead to achieve true trustworthy isomorphism-free generation. 
First of all, the isomorphisms are symmetries of the original problem, but not necessarily of the encoding used. 
Secondly, while VeriPB could be used to certify that symmetry breaking does not remove all solutions, it cannot be used (unless the set of symmetries it is allowed to use as a witness is somehow forced to be precisely the set of symmetries of the original problem) to certify that \emph{at least one} representative of each isomorphic class is preserved.
Finally, it cannot certify that the performed symmetry breaking was \emph{complete}, \ie that no duplicate solutions were enumerated. 
Providing true proof logging for isomorphism-free generation appears to be a major challenge. 

In the future, the methods we used can be extended to the construction of other
combinatorial structures similar to those considered here. This includes racks
and quandles, which are used in topology to construct invariants of knots;
arbitrary solutions (e.g., non-involutive or with relaxed degeneracy
conditions); and objects that appear in algebraic logic, especially L-algebras.

\begin{credits}
	\subsubsection{\ackname} This work was partially supported by Fonds Wetenschappelijk Onderzoek -- Vlaanderen (projects G070521N and G004124N) and by the project OZR3762 of Vrije Universiteit Brussel.
	The work was also partially funded by the European Union (ERC, CertiFOX, 101122653). Views and opinions expressed are however those of the author(s) only and do not necessarily reflect those of the European Union or the European Research Council. Neither the European Union nor the granting authority can be held responsible for them.
	
	\subsubsection{\discintname}
	The authors have no competing interests to declare that are
	relevant to the content of this article. 
	\end{credits}

\bibliographystyle{splncs04}
\bibliography{bb_refs_nourl.bib,dvc_refs.bib}

\newpage
\begin{subappendices}
\renewcommand{\thesection}{\arabic{section}}

\section{Encoding the cycle set properties}
\label{app:propCNF}
First, we encode the following properties:
\begin{enumerate}
 \item \label{prop:bij} for all $x,y,z\in\set$ with $y\neq z$, $\ybemat_{x,y}\neq\ybemat_{x,z}$ (i.e., the map $\leftmult_x$ is bijective for all $x\in\set$),
 \item \label{prop:cycloid} for all $x,y,z\in\set$, $\ybemat_{\ybemat_{x,y},\ybemat_{x,z}}=\ybemat_{\ybemat_{y,x},\ybemat_{y,z}}$ (i.e., the cycloid equation holds) and
 \item \label{prop:nondeg} for all $x,y\in\set$ with $x\neq y$, $\ybemat_{x,x}\neq\ybemat_{y,y}$ (i.e., the cycle set is non-degenerate).
\end{enumerate}
Properties \ref{prop:bij} and \ref{prop:nondeg} can straightforwardly be encoded using clauses that express that each number should occur on each row and on the diagonal: 
\begin{align}
	&\texttt{ExactlyOne}(\{ \matvar_{i,j,k}\mid j\in\set \}), &&\text{(for each $i,k\in \set$)}\\
	&\texttt{ExactlyOne}(\{ \matvar_{i,i,k}\mid i\in\set \}). &&\text{(for each $k\in \set$)}
\end{align}

The cycloid equation, property \ref{prop:cycloid}, needs some more care.
We achieve this using the clauses
\begin{align}
	& \olnot \matvar_{i,j,x}\lor\olnot \matvar_{i,k,y}\lor\olnot \matvar_{x,y,b}\lor\yvar_{i,j,k,b},&& \text{for all $i,j,k,x,y,b\in \set$ where $i<j$}\\
	& \olnot \matvar_{j,i,x}\lor\olnot \matvar_{j,k,y}\lor\olnot \matvar_{x,y,b}\lor\yvar_{i,j,k,b},&& \text{for all $i,j,k,x,y,b\in \set$ where $i<j$}\\
	&\texttt{ExactlyOne}(\{\yvar_{i,j,k,b}\mid b\in \set\}).&&\text{for all $i,j,k\in \set$ where $i<j$}
\end{align}
The first clause expresses that the variable $\yvar_{i,j,k,b}$ must be true whenever $\ybemat_{\ybemat_{i,j},\ybemat_{i,k}} = b$, the second does the same for the right-hand side of the cycloid equation. 
The totalizer used at the end enforces that this variable can only be true for one $b$, and as such the two sides of the equation must be the same.

\section{Proofs}
\label{app:proofs}
\begin{proposition}[\cref{prop:prop}, restated]
	Let $\pybemat$ and $\pybemat'$ be two partial cycle sets. The following properties hold.
\begin{enumerate}
 \item \label{prop:fst} If $\pybemat\lneqP\pybemat'$, then for all  $\ybemat\in\xcycsets{\pybemat}$ and $\ybemat'\in\xcycsets{\pybemat'}$, it holds that $\ybemat\lneqC\ybemat'$.
 \item \label{prop:snd} If $\pybemat\leqP\pybemat'$, then for all  $\ybemat\in\xcycsets{\pybemat}$ and $\ybemat'\in\xcycsets{\pybemat'}$, it holds that $\ybemat\leqC\ybemat'$. 
 \item \label{prop:thrd} If $\pybemat\leqP\pybemat'$ and $\pybemat'\leqP\pybemat$, then $\pybemat=\pybemat'$ and both cycle sets are complete.
\end{enumerate}
\end{proposition}
\begin{proof}
	Given two partial cycle sets $\pybemat,\pybemat'\in\pcycsets$, for which $\pybemat\leqP\pybemat'$, we show that Property \ref{prop:snd} holds.
	If $\pybemat\leqP\pybemat'$, for all cells $c\in\set\times\set$, where $\set=\{1,\ldots,n\}$, it holds that $\pybemat_c\leqP\pybemat'_c$, or equivalently that $\max{\pybemat_c}\leq\min{\pybemat'_c}$.
	Hence, for all cells $c$ of all extended cycle sets $\ybemat\in\xcycsets{\pybemat}$ and $\ybemat'\in\xcycsets{\pybemat'}$, it holds that $\ybemat_c$ is smaller than or equal to $\ybemat'_c$, or equivalently that $\ybemat\leqC\ybemat'$.
	
	If $\pybemat\lneqP\pybemat'$, there exists a cell $c$ for which $\pybemat_c\lneqP\pybemat'_c$ and for all cells $c'<c$, $\pybemat_{c'}\leqP\pybemat'_{c'}$.
	Hence, for $c'<c$, we have that $\max{\pybemat_{c'}}\leq\min{\pybemat'_{c'}}$, and hence that $\ybemat_{c'}\leq\ybemat'_{c'}$.
	However, $\max{\pybemat_{c}}$ is strictly smaller than $\min{\pybemat_{c'}}$, so $\ybemat_{c}$ will be strictly smaller than $\ybemat'_{c}$ for all extended cycle sets.
	As a result, we have that $\ybemat\lneqC\ybemat'$, and we have shown that Property \ref{prop:fst} holds.
	
	Last, we show that Property \ref{prop:thrd} holds as well.
	If $\pybemat\leqP\pybemat'$ and $\pybemat'\leqP\pybemat$, then for all cells $c$, we have that $\max{\pybemat_c}\leq\min{\pybemat'_c}$ and $\max{\pybemat'_c}\leq\min{\pybemat_c}$.
	Using that minimum values are smaller than or equal to maximum values gives us that \[\max{\pybemat_c}\leq\min{\pybemat'_c}\leq\max{\pybemat'_c}\leq\min{\pybemat_c}.\] 
	However, this implies that $\max{\pybemat_c}\leq\min{\pybemat_c}$, but we know that $\min{\pybemat_c}\leq\max{\pybemat_c}$.
	So necessarily, $\max{\pybemat_c}=\min{\pybemat_c}$, and hence, \[\max{\pybemat_c}=\min{\pybemat_c}=\max{\pybemat'_c}=\min{\pybemat'_c}.\]
	In other words, all cells in $\pybemat$ and $\pybemat'$ are fixed such that $\pybemat=\pybemat'$, which gives us Property \ref{prop:thrd}.
\end{proof}

\begin{theorem}[\cref{thm:main}, restated]
	Given a partial cycle set $\pybemat$ and permutation $\perm$, we have that:
	\begin{itemize}
	 \item if $\perm(\pybemat)\lneqP\pybemat$, then $\perm$ is a witness of non-minimality, 
	 \item If $\perm(\pybemat)\leqP_c\pybemat$ then
	 there is an extension $\pybemat'$ of $\pybemat$ such that
	 \begin{enumerate}
		 \item $\pybemat'$ and $\perm(\pybemat')$ are fully defined on all $c'\leq c$, and  
		 \item all $\leqC$-minimal extensions $\ybemat\in\xcycsets{\pybemat}$ are also extensions of $\pybemat'$.
	 \end{enumerate}
	\end{itemize}
   \end{theorem}
\begin{proof}
Let $\pybemat\in\pcycsets$ be a partial cycle set, and $\perm$ a permutation such that $\perm(\pybemat)\lneqP\pybemat$.
Using Property \ref{prop:fst} from the previous proposition, we obtain that $\perm(\ybemat)\lneqC\ybemat$ for all extended cycle sets $\ybemat\in\xcycsets{\pybemat}$.
As a result we have shown that $\perm$ is indeed a witness of non-minimality.
 
We show that the second case holds as well.
Let $\pybemat$ be a partial cycle set that is not fully defined and $\perm$ a permutation such that $\perm(\pybemat)\leqP_c\pybemat$.
For the first cell $c'\leq c$, where either, or both $\pybemat_{c'}$ and $\perm(\pybemat)_{c'}$ are not fully defined, the only way to avoid that $\perm(\pybemat)_{c'} \lneqP \pybemat_{c'}$ (in which case this would be a witness of non-minimality, following the first point) is assigning $\perm(\pybemat)_{c'}$ its maximal value and $\pybemat_{c'}$ its minimal value. We can repeat this for all cells up to $c$. 
\end{proof}

\section{Encoding the minimality check}
\label{app:mincheckCNF}
First, we describe the clauses used to ensure that the permutation $\pi$ is indeed well-defined and an isomorphism of the given cycle set.
In order to ensure that $\pi$ is well-defined, we add the following clauses:
\begin{align}
	&\texttt{ExactlyOne}(\{p_{i,j}\mid j\in\set \}), &&\text{(for each $i\in \set$)}\\
	&\texttt{ExactlyOne}(\{p_{i,j}\mid i\in\set \}). &&\text{(for each $j\in \set$)}.
\end{align}
Next, we ensure that $\pi$ is indeed an isomorphism of the problem.
Similar to the propagation step during the backtracking approach, no extra measures need to be taken if the diagonal equals the identity.
In all other cases, we need to ensure that $\pi$ fixes the diagonal $T$, or in other words that $\perm\in\centr_{\symmgroup_{|\set|}}(T)$.
The permutation $\perm$ fixes the diagonal if it maps each cycle of the diagonal onto a same-length cycle, in the same order.
Hence, if $\perm(i)=j$, the encoding ensures that all values $k$ in the same cycle as $i$ are fixed accordingly.
Using $n(k)$ to denote the successor of $k$ in its cycle and $C_n$ to denote the set containing all elements in cycles of length $n$, we add the following clauses to ensure that the entire cycle is fixed if $\perm(i)=j$; for all $n\in\set$ and all $i,j\in C_n$:
\begin{align}
 &\lnot p_{i,j} \lor p_{n(i),n(j)}.
\end{align}
	
To ensure that $\mat'$ is the image of $\mat$ under the given permutation $\pi$, we add a set of clauses relating the two matrices and the permutation.
We define the following clauses that ensure $\mat'$ equals $\pi(\mat)$; 
for all $i,i',j,j',k,k'\in\set$:
\begin{align}
	&\lnot p_{i,i'}\lor\lnot p_{j,j'}\lor\lnot p_{k,k'}\lor\lnot \omatvar_{i',j',k'}\lor \omatvar'_{i,j,k},\\
	&\lnot p_{i,i'}\lor\lnot p_{j,j'}\lor\lnot p_{k,k'}\lor\omatvar_{i',j',k'}\lor\lnot \omatvar'_{i,j,k}
\end{align}
This encoding can be simplified using the available domain-specific knowledge. 
The permutation $\pi$ is an isomorphism of the problem if it fixes the diagonal, hence, the image of $i\in\set$ under $\pi$ should equal an element $j\in\set$ that belongs to a cycle with the same length.
As a result, $\pi_{i,j}$ is false if $i$ and $j$ belong to cycles with different lengths, so this variable is no longer introduced and clauses containing it can be simplified (or simply removed).

\newcommand\plusplus{{+}{+}}
To ensure that $\mat'<\mat$, we use the order $\leq_{\omatvar'}$ and auxiliary variables $n_{c,i}$ defined in \cref{ss:MCE} to add static symmetry breaking constraints.
We use $(c,k)_{\plusplus}$ to denote the tuple $(c',k')$, where $\omatvar'_{c,k}$ immediately precedes $\omatvar'_{c',k'}$ in the order $\leq_{\omatvar'}$,
and add the following clauses for all cells $c$ and for each $k\in \set$ (except for the last inequality): 
\begin{align}
    &\nbvar_{c,k}\Rightarrow \omatvar_{(c,k)_{\plusplus}} \lor\lnot \omatvar'_{(c,k)_{\plusplus}}\\
    &\nbvar_{c,k}\land \omatvar'_{(c,k)_{\plusplus}}\Rightarrow \nbvar_{(c,k)_{\plusplus}}\\
    &\nbvar_{c,k}\land \lnot \omatvar_{(c,k)_{\plusplus}}\Rightarrow \nbvar_{(c,k)_{\plusplus}}.
\end{align}
The first clause forces $\mat'$ to be smaller than or equal to $\mat$, the last two clauses ensure that $\nbvar_{c,k}$ is only true if all preceding $\omatvar^{(')}$-variables have the same truth value.
In order to ensure that $\mat'$ is strictly smaller than $\mat$, we ensure that the last inequality is strict:
\begin{align}
    &\nbvar_{c_l,2}\Rightarrow \omatvar_{c_l,1} \lor\lnot \omatvar'_{c_l,1}\\
    &\lnot \nbvar_{c_l,2}\lor \lnot \omatvar'_{c_l,1}\\
    &\lnot \nbvar_{c_l,2}\lor \omatvar_{c_l,1},
\end{align}
where $c_l$ is the last cell.

In order to reason about partial cycle sets, a different order is used.
We no longer verify whether $\mat'<\mat$, but whether $\mat'\lneqP \mat$, or in other words, whether there exists a cell $c$ such that $\max{\mat'_c}<\min{\mat_c}$ and $\mat'_{c'}=\mat_{c'}$ for all previous cells $c'$.
As a result, the symmetry breaking constraints need to be adapted to reason about the maximum and minimum values of the cycle sets.
First, we encode the $l$- and $g$-variables by adding the following clauses to the formula for all $i,j,k\in\set$:
\begin{align}
	&g_{i,j,k}\Leftrightarrow \bigwedge_{l\in\set,l\leq k}\lnot \omatvar_{i,j,k}\\
	&l_{i,j,k}\Leftrightarrow \bigwedge_{l\in\set,l\geq k}\lnot \omatvar'_{i,j,k}.
\end{align}
Next, we add the following clauses for all cells $c$ and values $k\in\set$:
\begin{align}
    &\nbvar'_{c,k}\Rightarrow g_{c,k-1}\lor l_{c,k}\\
    &\nbvar'_{c,k}\land \lnot \omatvar_{(c,k)_{\plusplus}}\Rightarrow \nbvar'_{(c,k)_{\plusplus}}\\
    &\nbvar'_{c,k}\land \lnot l_{(c,k)_{\plusplus}}\Rightarrow \nbvar'_{(c,k)_{\plusplus}}.
\end{align}
In order to ensure that $\mat'\lneqP\mat$, we again ensure that the last inequality is strict:
\begin{align}
    &\nbvar'_{c_l,2}\Rightarrow g_{c_l,1}\lor l_{c_l,2}\\
    &\lnot \nbvar'_{c_l,2}\lor \omatvar_{c_l,1}\\
    &\lnot \nbvar'_{c_l,2}\lor l_{c_l,1},
\end{align}
where $c_l$ is the last cell.

\section{Configurations}
\label{app:conf}
We have used preliminary experiments enumerating sizes 8, 9 and 10 in order to identify good overall parameters for both \texttt{YBE-SMS} approaches.
For the backtracking approach, this eventually resulted in a situation where almost all the time is spent in unavoidable minimality checks (of complete, minimal solutions), and hence further parameter tweaking cannot result in significant speedups. 
For the incremental approach, the minimality check is no longer dominating the cost, so further optimization of these parameters (or other ideas for speeding up the search of the main solver) could be beneficial.
Both approaches use the following optimizations discussed throughout the paper:
\begin{itemize}
 \item they partition the problem by fixing the diagonals,
 \item they optimize the encoding(s) by using the information gained by fixing a diagonal,
 \item they encode \texttt{ExactlyOne} constraints using a binary encoding.
\end{itemize}
Next, we identified parameters that find a suitable balance between the frequency of and the limitations we impose on the minimality check:
\begin{itemize}
 \item Backtracking approach:
 \begin{itemize}
  \item Frequency: $\frac{1}{50}$
  \item Limit: 200 visited nodes in the search tree
 \end{itemize}
 \item Incremental approach:
 \begin{itemize}
  \item Frequency: $\frac{1}{100}$
  \item Limit: 10 conflicts
 \end{itemize}
\end{itemize}

\section{Experiments}
\label{app:expr}
In \cref{tab:gapvsssms} we give a more detailed breakdown of the results presented in \cref{fig:gapvssms}.
Here, we compare the new \texttt{YBE-SMS} implementations against the implementation of Akg\"un, Mereb and Vendramin~\cite{AMV22Enumerationset-theoreticsolutionsYang-Baxterequation} (referred to as AMV22).
\begin{table}[]
    \centering
    \resizebox{0.9\textwidth}{!}{
    \begin{tabular}{ll|ll|lll|lll|l}
    \cline{3-10}
                              &              & \multicolumn{2}{c|}{\textbf{AMV22}} & \multicolumn{3}{c|}{\textbf{Backtracking Approach}} & \multicolumn{3}{c|}{\textbf{Incr. SAT Approach}}   &  \\ \cline{1-10}
    \multicolumn{1}{|l}{Size} & \# Sols      & Iso Check (s.)         & Total (s.)         & MinCheck (s.)  & Total (s.)  & Speedup              & MinCheck (s.) & Total (s.)  & Speedup              &  \\ \cline{1-10}
    \multicolumn{1}{|l}{2}    & 2            & 0.0                    & 2.8                & 0.0            & \textbf{0.0}         &  & 0.0           & \textbf{0.0}         &  &  \\
    \multicolumn{1}{|l}{3}    & 5            & 0.0                    & 2.7                & 0.0            & \textbf{0.0}         &  & 0.0           & \textbf{0.0}         &  &  \\
    \multicolumn{1}{|l}{4}    & 23           & 0.0                    & 5.2                & 0.0            & \textbf{0.0}         &  & 0.0           & \textbf{0.0}         &  &  \\
    \multicolumn{1}{|l}{5}    & 88           & 0.0                    & 9.5                & 0.0            & \textbf{0.0}         &  & 0.0           & \textbf{0.0}         &  &  \\
    \multicolumn{1}{|l}{6}    & 595          & 0.2                    & 32.2               & 0.1            & \textbf{0.2}         & 161.0                     & 0.3           & 0.7         & 46.0                 &  \\
    \multicolumn{1}{|l}{7}    & 3\,456       & 1.1                    & 89.8               & 0.7            & \textbf{1.8}         & 49.9                      & 1.9           & 4.4         & 20.4                 &  \\
    \multicolumn{1}{|l}{8}    & 34\,530      & 43.1                   & 419.3              & 19.3           & \textbf{32.6}        & 12.9                      & 24.6          & 49.7        & 8.4                 &  \\
    \multicolumn{1}{|l}{9}    & 321\,931     & 2\,542.3               & 7\,797.7           & 621.6          & 760.5       & 10.2                      & 185.6         & \textbf{421.2}       & 18.51                 &  \\
    \multicolumn{1}{|l}{10}   & 4\,895\,272  & 237\,307.1             & 720\,883.0         & 41\,594.1      & 44\,792.5   & 16.1                      & 2\,706.3      & \textbf{6796.8}      & 108.1                 &  \\
    \multicolumn{1}{|l}{11}   & 77\,182\,093 &                        &                    &                &             &                           & 50\,767.2     & \textbf{226\,395.6} &  &  \\ \cline{1-10} \\
    \end{tabular} 
    }
    \caption{Comparing the runtimes of the implementation of AMV22 and our approaches building on SAT Modulo Symmetries.}
	\label{tab:gapvsssms}
    \end{table}
\end{subappendices}

\end{document}

\typeout{get arXiv to do 4 passes: Label(s) may have changed. Rerun}

%% file: bb_common-citations.tex
\usepackage{etoolbox}

\ExplSyntaxOn

\newcommand\setcitation[2]{%
	\csdef{mycommoncitation\text_uppercase:n{#1}}{#2}%
	\csappto{bbAllCommonCitations}{\cite{#2}\ }
}
\newcommand\getcitation[1]{%
	\csuse{mycommoncitation\text_uppercase:n{#1}}}

\newcommand\refto[1]{%
	\ifcsname  mycommoncitation\text_uppercase:n{#1}\endcsname%
	\getcitation{#1}%
	\else%
	#1%
	\fi%
}

\ExplSyntaxOff

\newcommand\mycite[1]{\cite{\refto{#1}}}

\setcitation{IDP}{DBBJD18PredicatelogicmodelinglanguageIDPsystem}
\setcitation{fodot}{DT08logicnonmonotoneinductivedefinitions}
\setcitation{CPsupport}{DBDD13ModelExpansionPresenceFunctionSymbolsUsing}
\setcitation{minisatid}{DBDD13ModelExpansionPresenceFunctionSymbolsUsing}
\setcitation{FunctionDetection}{DB13Detectionexploitationfunctionaldependenciesmodelgeneration}
\setcitation{lazyGrounding}{DDBS15LazyModelExpansionInterleavingGroundingSearch} 
\setcitation{Bootstrapping}{BJDJBD16BootstrappingInferenceIDPKnowledgeBaseSystem}
\setcitation{BootstrappingIDP}{BJDJBD16BootstrappingInferenceIDPKnowledgeBaseSystem}
\setcitation{FOSymmetry}{DBBD16localdomainsymmetrymodelexpansion}
\setcitation{FOLASP}{VDV21FOLASPFOInputLanguageAnswerSet}
\setcitation{inputster}{JJJ13CompilingInputFOinductivedefinitionsinto}
\setcitation{LearningPaper}{BBBDDJLR15Predicatelogicmodelinglanguagemodelingsolving}

\setcitation{foid}{DT08logicnonmonotoneinductivedefinitions}
\setcitation{cplogic}{VDB10FOIDextensionDLrules}
\setcitation{clog}{BVDV14FOCKnowledgeRepresentationLanguageCausality} 
\setcitation{foc}{BVDV14FOCKnowledgeRepresentationLanguageCausality}
\setcitation{inferenceClog}{BVDV14InferenceFOCModellingLanguage}
\setcitation{inferenceFOC}{BVDV14InferenceFOCModellingLanguage}
\setcitation{examplesClog}{BVDV14FOCRelatedModellingParadigms} 
\setcitation{satid}{MWDB08SATIDSatisfiabilityPropositionalLogicExtendedInductive}

\setcitation{fodot2asp}{DLTV19informalsemanticsAnswerSetProgrammingTarskian} 
\setcitation{Tarskian}{DLTV19informalsemanticsAnswerSetProgrammingTarskian} 
\setcitation{TarskianSemanticsASP}{DLTV19informalsemanticsAnswerSetProgrammingTarskian} 
\setcitation{KBS}{DV08BuildingKnowledgeBaseSystemIntegrationLogic}
\setcitation{KBS-invitedtalk}{D16FOKnowledgeBaseSystemproject}
\setcitation{KBPE}{DWD11PrototypeKnowledge-BasedProgrammingEnvironment}

\setcitation{justifications}{DBS15FormalTheoryJustifications} 
\setcitation{justificationTheory}{DBS15FormalTheoryJustifications} 
\setcitation{AFT}{DMT00ApproximationsStableOperatorsWell-FoundedFixpointsApplications}

\setcitation{SAT}{BHvW21HandbookSatisfiability-SecondEdition}
\setcitation{HandbookOfSAT}{BHvW21HandbookSatisfiability-SecondEdition}
\setcitation{SMS}{KS21SATModuloSymmetriesGraphGeneration}

\setcitation{totalizer}{BB03EfficientCNFEncodingBooleanCardinalityConstraints}
\setcitation{saucy}{KSM10SymmetrySatisfiabilityUpdate}

\setcitation{DPLLT}{GHNOT04DPLLTFastDecisionProcedures}
\setcitation{SMT}{BSST21SatisfiabilityModuloTheories}

\setcitation{CP}{RvW06HandbookConstraintProgramming}
\setcitation{csp2asp}{DW11Translation-BasedConstraintAnswerSetSolving}
\setcitation{EZCSP}{B11Industrial-SizeSchedulingASPCP}
\setcitation{Inca}{DW12AnswerSetSolvingLazyNogoodGeneration}
\setcitation{Zinc}{MNRSGW08DesignZincModellingLanguage}
\setcitation{MiniZinc}{NSBBDT07MiniZincTowardsStandardCPModellingLanguage}

\setcitation{ASPComp2}{DVBGT09SecondAnswerSetProgrammingCompetition}
\setcitation{ASPComp3}{CIR14thirdopenanswersetprogrammingcompetition}
\setcitation{ASPComp4}{ACCDDIKK13FourthAnswerSetProgrammingCompetitionPreliminary}
\setcitation{ASPComp5}{CGMR16DesignresultsFifthAnswerSetProgramming}
\setcitation{ASPComp6}{GMR17SixthAnswerSetProgrammingCompetition}
\setcitation{ASPComp7}{GMR20SeventhAnswerSetProgrammingCompetitionDesign}
\setcitation{AspInPractice}{GKKS12AnswerSetSolvingPractice}
\setcitation{clasp}{GKS12Conflict-drivenanswersetsolvingFromtheory}
\setcitation{oclingo}{GGKOSS12StreamReasoningAnswerSetProgrammingPreliminary}
\setcitation{clingo}{GKKOSW16TheorySolvingMadeEasyClingo5}
\setcitation{gringo}{GST07GrinGoNewGrounderAnswerSetProgramming}
\setcitation{cmodels}{GLM04SAT-BasedAnswerSetProgramming}
\setcitation{DLV}{LPFEGPS06DLVsystemknowledgerepresentationreasoning}
\setcitation{GroundingBottleneck}{SPL14TwoApplicationsASP-PrologSystemDecomposablePrograms}
\setcitation{lazygroundingASP}{BW18ExploitingJustificationsLazyGroundingAnswerSet} 
\setcitation{justificationsAlpha}{BW18ExploitingJustificationsLazyGroundingAnswerSet}
\setcitation{ASP}{MT99StableModelsAlternativeLogicProgrammingParadigm}
\setcitation{OLL}{AKMS12Unsatisfiability-basedoptimizationclasp} 

\setcitation{MaxSAT}{LM21MaxSATHardSoftConstraints}

\setcitation{KR}{vLP08HandbookKnowledgeRepresentation}

\setcitation{lazyclausegeneration}{OSC09Propagationlazyclausegeneration}
\setcitation{FP}{H89ConceptionEvolutionApplicationFunctionalProgrammingLanguages}
\setcitation{GroundingWithBounds}{WMD10GroundingFOFOIDBounds}
\setcitation{GroundWithBounds}{WMD10GroundingFOFOIDBounds}
\setcitation{LTC}{BJBDVD14SimulatingDynamicSystemsUsingLinearTime}
\setcitation{SPSAT}{DBDDM12SymmetryPropagationImprovedDynamicSymmetryBreaking}
\setcitation{BreakID}{DBBD16ImprovedStaticSymmetryBreakingSAT}
\setcitation{LCG}{S10LazyClauseGenerationCombiningPowerSAT}

\setcitation{GroundedFixpoints}{BVD15Groundedfixpointstheirapplicationsknowledgerepresentation}
\setcitation{PartialGroundedFixpoints}{BVD15PartialGroundedFixpoints}
\setcitation{LogicBlox}{GAK12LogicBloxPlatformLanguageTutorial}
\setcitation{proB}{LB08ProBautomatedanalysistoolsetBmethod}
\setcitation{NaturalInductions}{DV14Well-FoundedSemanticsPrincipleInductiveDefinitionRevisited} 
\setcitation{LP}{vK76SemanticsPredicateLogicProgrammingLanguage}

\setcitation{AF}{D95AcceptabilityArgumentsFundamentalRoleNonmonotonicReasoning}
\setcitation{ADF}{BW10AbstractDialecticalFrameworks}
\setcitation{ADFRevisited}{BSEWW13AbstractDialecticalFrameworksRevisited}
\setcitation{DefaultLogic}{R80LogicDefaultReasoning}
\setcitation{DL}{R80LogicDefaultReasoning}
\setcitation{AEL}{M85SemanticalConsiderationsNonmonotonicLogic}
\setcitation{minisat}{ES03ExtensibleSAT-solver}
\setcitation{completion}{C77NegationFailure}
\setcitation{ClarkCompletion}{C77NegationFailure}
\setcitation{wasp}{ADFLR13WASPNativeASPSolverBasedConstraint}
\setcitation{CEGAR}{CGJLV03Counterexample-guidedabstractionrefinementsymbolicmodelchecking}
\setcitation{CuttingPlane}{DFJ54SolutionLarge-ScaleTraveling-SalesmanProblem}
\setcitation{CuttingPlanes}{DFJ54SolutionLarge-ScaleTraveling-SalesmanProblem}
\setcitation{kodkod}{TJ07KodkodRelationalModelFinder}
\setcitation{CDCL}{SS99GRASPSearchAlgorithmPropositionalSatisfiability}
\setcitation{1UIP}{ZMMM01EfficientConflictDrivenLearningBooleanSatisfiability}
\setcitation{relevance}{JBDJD16RelevanceSATID}
\setcitation{relevance-implementation}{JDBJD16ImplementingRelevanceTrackerModule}
\setcitation{WFS}{VRS91Well-FoundedSemanticsGeneralLogicPrograms}
\setcitation{UnfoundedSet}{VRS91Well-FoundedSemanticsGeneralLogicPrograms}
\setcitation{UFS}{VRS91Well-FoundedSemanticsGeneralLogicPrograms}
\setcitation{stablesemantics}{GL88StableModelSemanticsLogicProgramming}
\setcitation{shatter}{ASM06EfficientSymmetryBreakingBooleanSatisfiability}
\setcitation{sbass}{DTW11Symmetry-breakinganswersetsolving}
\setcitation{lparsemanual}{url:lparsemanual}
\setcitation{AIC}{FGZ04Activeintegrityconstraints}
\setcitation{templates}{DvJD15Semanticstemplatescompositionalframeworkbuildinglogics}
\setcitation{templates2}{DvBJD16CompositionalTypedHigher-OrderLogicDefinitions}
\setcitation{sat-to-sat}{JTT16SAT-to-SATDeclarativeExtensionSATSolversNew}
\setcitation{sat-to-sat-qbf}{BJT16SolvingQBFInstancesNestedSATSolvers}
\setcitation{sat-to-sat-SO}{BJT16DeclarativeSolverDevelopmentCaseStudies}
\setcitation{XSB}{SW12XSBExtendingPrologTabledLogicProgramming}
\setcitation{KCmap}{DM02KnowledgeCompilationMap}
\setcitation{KnowledgeCompilationMap}{DM02KnowledgeCompilationMap}
\setcitation{TLA}{L02SpecifyingSystemsTLALanguageToolsHardware}
\setcitation{EventB}{A10ModelingEvent-B-SystemSoftwareEngineering}
\setcitation{MX}{MT05FrameworkRepresentingSolvingNPSearchProblems}
\setcitation{MIP}{S96Linearintegerprogramming-theorypractice}
\setcitation{PerfectModel}{P88DeclarativeSemanticsDeductiveDatabasesLogicPrograms}
\setcitation{SafeInductions}{BVD17SafeInductionsAlgebraicStudy}
\setcitation{alpha}{W17BlendingLazy-GroundingCDNLSearchAnswer-SetSolving}
\setcitation{omiga}{DEFWW12OMiGAOpenMindedGroundingOn-The-FlyAnswer}
\setcitation{gasp}{PDPR09GASPAnswerSetProgrammingLazyGrounding}
\setcitation{asperix}{LN09FirstVersionNewASPSolverASPeRiX}
\setcitation{CTL}{CE81DesignSynthesisSynchronizationSkeletonsUsingBranching-Time}
\setcitation{AFT-AIC}{BC18Fixpointsemanticsactiveintegrityconstraints}
\setcitation{UltimateApproximator}{DMT04Ultimateapproximationapplicationnonmonotonicknowledgerepresentation}
\setcitation{KripkeKleene}{F85Kripke-KleeneSemanticsLogicPrograms}
\setcitation{AFT-HO}{CRS18ApproximationFixpointTheoryWell-FoundedSemanticsHigher-Order} 
\setcitation{HereThere}{H30formalenRegelnintuitionistischenLogik}
\setcitation{dAEL}{VCBD16DistributedAutoepistemicLogicApplicationAccessControl}
\setcitation{SDD}{D11SDDNewCanonicalRepresentationPropositionalKnowledge}
\setcitation{HEX}{EIST06dlvhexProverSemantic-WebReasoningunderAnswer-Set}
\setcitation{wADF}{BSWW18WeightedAbstractDialecticalFrameworks}
\setcitation{wADFfix}{BPSWW18WeightedAbstractDialecticalFrameworksExtendedRevised}
\setcitation{TransitionSystems}{NOT06SolvingSATSATModuloTheoriesFrom}
\setcitation{galliwasp}{MG12GalliwaspGoal-DirectedAnswerSetSolver}
\setcitation{clingcon}{BKOS17Clingconnextgeneration}
\setcitation{lp2sat}{JN11CompactTranslationsNon-disjunctiveAnswerSetPrograms}
\setcitation{lp2mip}{LJN12AnswerSetProgrammingMixedIntegerProgramming}
\setcitation{lp2diff}{JNS09ComputingStableModelsReductionsDifferenceLogic}
\setcitation{lp2acyc}{GJR14AnswerSetProgrammingSATmoduloAcyclicity}
\setcitation{PB}{RM09Pseudo-BooleanCardinalityConstraints}
\setcitation{CuttingPlanes}{CCT87complexitycutting-planeproofs}
\setcitation{RoundingSAT}{EN18DivideConquerTowardsFasterPseudo-BooleanSolving}
\setcitation{PRS}{DG02InferenceMethodsPseudo-BooleanSatisfiabilitySolver}
\setcitation{sat4j}{LP10Sat4jlibraryrelease22}
\setcitation{mingo}{LJN12AnswerSetProgrammingMixedIntegerProgramming}
\setcitation{pbmodels}{LT05Pbmodels-SoftwareComputeStableModels}
\setcitation{HEF-LP}{GLL07Head-Elementary-Set-FreeLogicPrograms}
\setcitation{HCF-LP}{BD94PropositionalSemanticsDisjunctiveLogicPrograms}
\setcitation{aspcore2}{CFGIKKLM20ASP-Core-2InputLanguageFormat}
\setcitation{SemanticWeb}{AGvH12SemanticWebPrimer3rdEdition}
\setcitation{QBF}{KB21TheoryQuantifiedBooleanFormulas}
\setcitation{SMUS}{IPLM15SmallestMUSExtractionMinimalHittingSet}
\setcitation{CDCLsym}{MBCK18CDCLSymIntroducingEffectiveSymmetryBreakingSAT}
\setcitation{SEL}{DBB17SymmetricExplanationLearningEffectiveDynamicSymmetry}
\setcitation{Coq}{BC04InteractiveTheoremProvingProgramDevelopment-}
\setcitation{Isabelle}{NPW02IsabelleHOL-ProofAssistantHigher-OrderLogic}
\setcitation{HOL}{SN08BriefOverviewHOL4}
\setcitation{lean}{dU21Lean4TheoremProverProgrammingLanguage}
\setcitation{CakeML}{TMKFON19verifiedCakeMLcompilerbackend}
\setcitation{FRAT}{BCH22FlexibleProofFormatSATSolver-ElaboratorCommunication}
\setcitation{DRUP}{GN03VerificationProofsUnsatisfiabilityCNFFormulas}
\setcitation{RUP}{GN03VerificationProofsUnsatisfiabilityCNFFormulas}
\setcitation{GRIT}{CMS17EfficientCertifiedResolutionProofChecking}
\setcitation{TraceCheck}{url:TraceCheck}
\setcitation{LRAT}{CHHKS17EfficientCertifiedRATVerification}

\setcitation{PR}{HKB17ShortProofsWithoutNewVariables}
\setcitation{DRAT}{WHH14DRAT-trimEfficientCheckingTrimmingUsingExpressive}
\setcitation{satcomp2016}{BHJ17SATCompetition2016RecentDevelopments}
\setcitation{satcomp2013}{url:SATcomp2013}

\setcitation{Essence}{FHJHM08Essenceconstraintlanguagespecifyingcombinatorialproblems}
\setcitation{smodels}{SN01SmodelsSystem}
\setcitation{lparse}{S98ImplementationLocalGroundingLogicProgramsStable}
\setcitation{IDPZ3}{CVVD22IDP-Z3reasoningengineFO} 
\setcitation{IDP-Z3}{CVVD22IDP-Z3reasoningengineFO} 
\setcitation{Z3}{dB08Z3EfficientSMTSolver}

\setcitation{DMN}{DMN} 
\setcitation{Planning}{GNT04Automatedplanning-theorypractice}
\setcitation{AutomatedPlanning}{GNT04Automatedplanning-theorypractice}
\setcitation{RC2}{IMM19RC2EfficientMaxSATSolver}
\setcitation{enfragmo}{phd/Aavani14}
\setcitation{TPTP}{url:tptp}

%% file: local-macros.tex
\newcommand\olnot[1]{\overline{#1}}
\newcommand\ltrue{\mathbf{t}}
\newcommand\lfalse{\mathbf{f}}
 

%% file: bb_common-tex-defs.tex
\usepackage{xspace}
\usepackage{amssymb}
\usepackage{amsmath}
\usepackage{ifthen}
\usepackage[dvipsnames]{xcolor}
\usepackage{url}

\newcommand{\ignore}[1]{}

\newboolean{nocomments}
\setboolean{nocomments}{false}

\newboolean{comments-in-margin}
\setboolean{comments-in-margin}{false}

\newcommand{\namedcomment}[3]{%
	\ifthenelse{\boolean{nocomments}}%
	{\relax}
	{
		\ifthenelse{\boolean{comments-in-margin}}%
		{ \marginpar{\color{#3}{\sc #2} #1}}  
		{  {\color{#3} \textsc{ #2}: #1}  } 
	}%
}

\newcommand{\bart}[1]{\namedcomment{#1}{bb}{Maroon}}

\newcommand{\eg}{e.g., }
\newcommand{\ie}{i.e., }

\ignore{

\usepackage{ifthen}
\usepackage{url}
\usepackage[usenames,dvipsnames,tikz]{xcolor}
\usepackage{amssymb}
\usepackage{amsmath}
\usepackage{import}
\usepackage{xparse}
\usepackage{amsmath}
\usepackage{xspace}
\usepackage{datatool}
\usepackage{glossaries}

\providecommand\m[1]{\ensuremath{#1}\xspace}
\renewcommand{\m}[1]{\ensuremath{#1}\xspace}
\newcommand{\trval}[1]{\m{\mathbf{#1}}}

	\newcommand{\lrule}{\leftarrow}
	\newcommand{\cause}{\stackrel{c}{\lrule}}
	
	\newcommand{\ltrue}{\trval{t}}
	\newcommand{\lfalse}{\trval{f}}

	\newcommand{\struct}{\m{I}}

	\newcommand{\set}{\m{S}}
	\NewDocumentCommand\inter{g+g}{%
	  \IfNoValueTF{#1}
	    {\struct}
	    {\m{#1^{#2}}}}

	\renewcommand{\int}{\m{\mathbb{Z}}}

	\NewDocumentCommand\subs{g+g}{%
	  \IfNoValueTF{#1}
	    {\m{/}}
	    {\m{#1/ #2}}}

	\newcommand{\tuple}[1]{\m{\left \langle #1 \right \rangle }}

\newcommand{\ouracronym}[3]{%
	\newacronym{#1}{#2}{#3}
	\expandafter\newcommand\csname #1\endcsname{\gls{#1}\xspace}%
}
	\ouracronym{FO}{FO}{first-order logic}
	\ouracronym{PC}{PC}{propositional calculus}
	\ouracronym{MX}{MX}{Model Expansion}
	\ouracronym{MO}{MO}{Model Optimization}
	\ouracronym{ASP}{ASP}{Answer Set Programming}
	\ouracronym{TP}{TP}{Theorem Proving}
	\ouracronym{LP}{LP}{Logic Programming}
	\ouracronym{CP}{CP}{Constraint Programming}
	\ouracronym{FP}{FP}{Functional Programming}
	\ouracronym{KR}{KR}{Knowledge Representation}
	\ouracronym{CSP}{CSP}{Constraint Satisfaction Problem}
	\ouracronym{SMT}{SMT}{SAT Modulo Theories}
	\ouracronym{KBS}{KBS}{knowledge base system}
	\ouracronym{NNF}{NNF}{Negation Normal Form}
	\ouracronym{FNNF}{FNNF}{Flat Negation Normal Form}
	\ouracronym{DefNNF}{DefNNF}{Definition Negation Normal Form}
	\ouracronym{DEFNF}{DEFNF}{Definition Normal Form}
	\ouracronym{CDCL}{CDCL}{Conflict-Driven Clause-Learning}
	\ouracronym{WFS}{WFS}{Well-Founded Semantics}
	\ouracronym{LCG}{LCG}{Lazy Clause Generation}
	\ouracronym{AEL}{AEL}{Autoepistemic Logic}
	\ouracronym{OEL}{OEL}{Ordered Epistemic Logic}
	\ouracronym{AFT}{AFT}{Approximation Fixpoint Theory}

	\makeatletter
	\def\ifenv#1{
	\def\@tempa{#1}%
	\def\@ttempa{#1*}%
	\ifx\@tempa\@currenvir
	\expandafter\@firstoftwo
	\else
	\expandafter\@secondoftwo
	\fi
	}
	\makeatother

	\newcommand{\ddrule}[4]{\ensuremath{#1 \leftarrow #2 & \{#3\} & #4}}
	\newcommand{\drule}[2]{\ensuremath{#1 & \leftarrow & #2}}

	\newcommand{\darule}[4]{\ensuremath{#1 \leftarrow #2 & \{#3\} & #4}}
	\newcommand{\arule}[2]{\ensuremath{#1 \, &\leftarrow \, #2}}

	\newcommand{\LNDRule}[2]{
	\ifenv{array}
	{\drule{#1}{#2}}
	{ \ifenv{align}
		{\arule{#1}{#2}}
		{\ifenv{align*}
		{\arule{#1}{#2}}
		{ERROR: using LDRule in unsupported environment: \@currenvir}
		}
	}
	}

	\newcommand{\LDRule}[4]{
	\ifenv{array}
	{\ddrule{#1}{#2}{#3}{#4}}
	{ \ifenv{align}
		{\darule{#1}{#2}{#3}{#4}}
		{\ifenv{align*}
		{\darule{#1}{#2}{#3}{#4}}
		{ERROR: using LDRule in unsupported environment: \@currenvir}
		}
	}
	}

	\NewDocumentCommand\LRule{m+g+g+g}{%
		\IfNoValueTF{#2}%
		{#1.&}{%
		\IfNoValueTF{#3}
		{\LNDRule{#1}{#2.}}
		{\LDRule{#1}{#2.}{#3}{#4}}%
		}
	}

	\NewDocumentCommand\CLRule{m+g}{%
	\ifenv{array}
	{\cdrule{#1}{#2}}
	{ \ifenv{align}
		{\carule{#1}{#2}}
		{\ifenv{align*}
			{\carule{#1}{#2}}
			{ERROR: using CLRule in unsupported environment: \@currenvir}
		}
	}
	}

	\NewDocumentCommand\carule{m+g}{%
		\IfNoValueTF{#2}
			{\ensuremath{#1.}}
			{\ensuremath{#1 \, &\cause \, #2}}}
	\NewDocumentCommand\cdrule{m+g}{%
		\IfNoValueTF{#2}
			{\ensuremath{#1.}}
			{\ensuremath{#1 & \cause & #2}}}

	\newcommand{\algrule}[4]{
	\hbox{{#1}:}& 
	\quad #2 ~\longrightarrow~ #3 
	\hbox{~ if } #4\\
	}

	\newcommand{\AlgoRule}[4]{
	\ifenv{array}
	{\algrule{#1}{#2}{#3}{#4}}
		{ERROR: using AlgoRule in unsupported environment: \@currenvir}
	}

}